\documentclass[11pt,letterpaper]{article}

\usepackage[utf8]{inputenc}
\usepackage[english]{babel}

\usepackage[dvipsnames,usenames]{xcolor}
\usepackage[colorlinks=true,pdfpagemode=UseNone,urlcolor=RoyalBlue,linkcolor=RoyalBlue,citecolor=OliveGreen,pdfstartview=FitH]{hyperref}

\usepackage{tikz}
\usetikzlibrary{shapes.geometric}

\usepackage{subcaption}

\usepackage[margin=1in]{geometry}
\usepackage{amsfonts,amsmath,amsthm,thmtools}
\usepackage{nicefrac}

\declaretheorem{theorem}
\declaretheorem[sibling=theorem]{lemma}
\declaretheorem[sibling=theorem]{corollary}

\declaretheorem[sibling=theorem]{definition}
\declaretheorem[sibling=theorem]{claim}
\declaretheorem[sibling=theorem]{observation}

\usepackage{todonotes}

\usepackage{booktabs}
\usepackage{caption}
\usepackage[numbers, sort]{natbib}

\usepackage{tcolorbox}

\usepackage{enumitem}
\usepackage{cleveref}

\newcommand{\E}{\mathbf{E}}
\renewcommand{\Pr}{\mathbf{Pr}}

\newcommand{\opt}{\textsc{Opt}}

\newcommand{\ls}{\textsc{Ls}}

\newcommand{\Int}{\textsc{Int}}
\newcommand{\Ext}{\textsc{Ext}}
\newcommand{\cost}{\textsc{Cost}}

\newcommand{\Sref}[1]{\hyperref[#1]{\S\ref*{#1}}}

\newcommand{\adm}{\text{adm}}
\newcommand{\calK}{\mathcal{K}}
\newcommand{\calC}{\mathcal{C}}

\newcommand{\rmCostStays}{\mathrm{CostStays}}
\newcommand{\rmCostMoves}{\mathrm{CostMoves}}
\newcommand{\rmEstCostStays}{\mathrm{EstCostStays}}
\newcommand{\rmEstCostMoves}{\mathrm{EstCostMoves}}
\newcommand{\rmExpCostStays}{\mathrm{ExpCostStays}}
\newcommand{\rmExpCostMoves}{\mathrm{ExpCostMoves}}

\title{Combinatorial Correlation Clustering}
\author{
    Vincent Cohen-Addad
    \thanks{Google Research. Email: cohenaddad@google.com.}
    \and
    David Rasmussen Lolck
    \thanks{University of Copenhagen. Email: dalo@di.ku.dk.}
    \and
    Marcin Pilipczuk
    \thanks{University of Warsaw. Email: malcin@mimuw.edu.pl.}
    \and
    Mikkel Thorup
    \thanks{University of Copenhagen. Email: mthorup@di.ku.dk.}
    \and
    Shuyi Yan
    \thanks{University of Copenhagen. Email: shya@di.ku.dk.}
    \and
    Hanwen Zhang
    \thanks{University of Copenhagen. Email: hazh@di.ku.dk.}
}

\date{}

\begin{document}
\begin{titlepage}
\def\thepage{}
\thispagestyle{empty}
\maketitle
\begin{abstract}

    Correlation Clustering is a classic clustering objective arising in numerous machine learning and data mining applications. Given a graph $G=(V,E)$, the goal is to partition the vertex set into  clusters so as to minimize the number of edges between clusters plus the number of edges missing within clusters.

    The problem is APX-hard and the best
    known polynomial time approximation factor is 
    1.73 by Cohen-Addad, Lee, Li,  and Newman [FOCS'23].  They use an LP with 
    $|V|^{1/\epsilon^{\Theta(1)}}$ variables for some small $\epsilon$. However, due to the  practical relevance of correlation clustering, there has also been great interest in getting more efficient sequential and parallel algorithms. 

    The classic combinatorial \emph{pivot} algorithm of Ailon, Charikar and Newman [JACM'08] provides
    a 3-approximation in linear time. Like most other algorithms discussed here, this uses randomization.
   Recently, Behnezhad, Charikar, Ma and  Tan [FOCS'22] presented a $3+\epsilon$-approximate solution for solving
    problem in a constant number of rounds in the Massively Parallel Computation (MPC) setting. Very recently,  Cao, Huang, Su [SODA'24] provided a 2.4-approximation in a polylogarithmic number of rounds in the MPC model and in $\tilde{O}
    (|E|^{1.5})$ time in the classic sequential  setting. They asked whether it
    is possible to get a better than 3-approximation in near-linear time?

    We resolve this problem with an efficient
    combinatorial algorithm providing a drastically better approximation factor. It achieves a $\sim 2-2/13 < 1.847$-approximation in sub-linear ($\tilde O(|V|)$) sequential time or in sub-linear ($\tilde O(|V|)$) space in the streaming setting. In the MPC model, we give an algorithm using only a constant number of rounds that achieves a $\sim 2-1/8 < 1.876$-approximation.

\end{abstract}

\end{titlepage}

\thispagestyle{empty}
\tableofcontents 
\thispagestyle{empty}
\newpage
\setcounter{page}1

\section{Introduction}

Correlation clustering is a fundamental clustering objective
that models a large number of machine learning and data mining applications. Given a set of data elements, represented as vertices $V$ in a graph $G$, and a set of pairs of similar elements, represented as edges  $E$ in the graph, the goal is to find a partition of the vertex sets that minimizes the number of missing edges within the parts of
the partition plus the number of edges across the parts of the partition. An alternative formulation is that we want a graph consisting of disjoint cliques, minimizing the symmetric difference to the input graph. Below $n=|V|$ and $m=|E|$.

The problem was originally introduced by Bansal, Blum, and Chawla in the early 
2000~\cite{BBC04} and has since then found a large variety of applications 
ranging from finding clustering ensembles
\cite{bonchi2013overlapping}, duplicate detection
\cite{arasu2009large}, community mining \cite{chen2012clustering},
disambiguation tasks \cite{kalashnikov2008web}, to automated labelling
\cite{agrawal2009generating, chakrabarti2008graph} and many more.

Thanks to its tremendous modelling success, Correlation clustering has been widely
studied in the algorithm design, machine learning and data mining communities.
From a complexity standpoint, the problem was shown to be APX-Hard~\cite{CGW05} and
Bansal, Blum, and Chawla~\cite{BBC04} gave the first $O(1)$-approximation 
algorithm. The constant was later improved to 4 by Charikar, Guruswami, and Wirth~\cite{CGW05}. 

In a landmark paper, Ailon, Charikar and Newman~\cite{ACN08} gave two
very important pivot-based algorithms. First, they presented a \emph{combinatorial} -- in the sense that it does not rely on solving a linear program --  pivot algorithm that iteratively picks a random vertex uniformly at random in the set of unclustered vertices and clusters it with all its unclustered neighbors.
They showed that this algorithm achieves a 3 approximation (and their analysis is tight). Next, they improved the approximation factor to 2.5 using a standard linear program (LP) which was shown to have an integrality gap of at least 2.  They still
used a random pivot, but instead of creating a cluster containing all the unclustered neighbors of the pivot, they randomly assigned each vertex to the cluster of the pivot based on the LP solution.  

A better rounding of this LP of Ailon, Charikar and Newman~\cite{ACN08} was later presented by 
Chawla, Makarychev, Schramm and Yaroslavtev~\cite{CMSY15} who achieved a 2.06-approximation still relying on a pivot-based approach to round, hence coming close to a nearly-optimal rounding given its integrality gap of at least 2.

Since the integrality gap of this LP is at least 2, 2 has appeared as a strong 
barrier for approximating Correlation Clustering, and 3 has remained 
the best-known non-LP-based algorithm to this day. Recently, different 
LP formulations have been used to get better than 2 approximations. 
Cohen-Addad, Lee and Newman~\cite{CLN22} showed that the Sherali-Adams hierarchy helps bypass the standard LP
integrality gap of 2 by providing a $1.995$ approximation which was very recently improved to
$1.73$ by Cohen-Addad, Lee, Li and Newman~\cite{CLLN23+} using a new linear programming formulation and relying on the Sherali-Adams hierarchy too. However, these improvements have happened at an even
greater computational cost: While the pivot-rounding algorithm required to solve a linear program on
$n^2$ variables, the above two algorithms require $n^{1/\epsilon^{\Theta(1)}}$ variables and running time for some small $\epsilon$.

\subsubsection*{Efficient algorithms}

Motivated by the large number of applications in machine learning and 
data mining applications, researchers have put an intense effort in 
\emph{efficiently} approximating Correlation Clustering: From obtaining 
linear time algorithm~\cite{ACN08}, to dynamic algorithms~\cite{BehnezhadDHSS-FOCS19}, distributed (Map-Reduce 
or Massively Parallel Computation (MPC) models)~\cite{DBLP:conf/focs/BehnezhadCMT22,CLMNP21}, streaming~\cite{BehnezhadCMT-SODA23,ChakrabartyM-NIPS23,CambusKLPU-SODA24,DBLP:conf/innovations/Assadi022,CLMNP21}, or sublinear time~\cite{DBLP:conf/innovations/Assadi022}.
In most of these models, the LP-based approaches have remained 
unsuccessful and the state-of-the-art algorithm remains the combinatorial pivot algorithm. 
Arguably, the modularity and simplicity of the combinatorial pivot 
algorithm and its analysis have been key to obtaining these results.

Thus, the question of \emph{how fast and well and in which model} one can approximate Correlation Clustering better than a factor 3 has naturally emerged.
At one end of the spectrum, we have the above $1.73$ in time $n^{1/\epsilon^{\Theta(1)}}$. At the other end of the spectrum,
Cohen-Addad, Lattanzi, Mitrovic, Norouzi-Fard, Parotsidis, and Tarnawski~\cite{CLMNP21} and Assadi and Wang~\cite{DBLP:conf/innovations/Assadi022} showed
that Correlation Clustering can be approximated to an $O(1)$-factor in linear time (i.e.: $O(m)$ time) and in fact in sublinear
time (i.e.: $n \log^{O(1)} n$) using random
sampling of neighbors. However, the constant proven in these works is larger than 500.
In a very recent work,~\cite{CaoHS-SODA24} showed that one can approximately solve and round the standard LP in time $O(m^{1.5})$
and achieve a $2.4$-approximation, coming close to the linear time-bound. 
In the streaming model,~\cite{DBLP:conf/focs/BehnezhadCMT22} gave a $3+\epsilon$-approximation using $O(1/\epsilon)$ passes.
Later~\cite{BehnezhadCMT-SODA23} gave a 1-pass algorithm achieving a 5-approximation. Very recently, 
~\cite{CambusKLPU-SODA24} and~\cite{chakrabarty2023single} independently provided a 1-pass  $3+\epsilon$-approximation.
\cite{CaoHS-SODA24} thus asks:
\emph{How fast and in which models can one approximate Correlation Clustering within a factor 3 and beyond?}

\subsection{Our Results}

Our main contribution is a novel combinatorial 
algorithm that we show 
achieves a $2-2/13$-approximation.  Our new algorithm alternates between performing a local search and a certain flip. The flip takes all cut edges in the current solution and aims to make them internal, doing so by increasing the cost of cutting them in the objective for the next local search. One can think of this as a very systematic method to escape a bad local minimum, leading to strong theoretical guarantees.

We highlight that despite the fact that the number of local search
iterations could be close to linear in the worst-case, we show that with a minor loss in approximation factor, the whole local search can be implemented in near-linear total time.

In fact, using random sampling of neighbors,
we can implement it in sublinear time, 
improving over the state-of-the-art sublinear-time
$>500$-approximation~\cite{DBLP:conf/innovations/Assadi022,CLMNP21}. On top of that, we can implement it in the streaming model. A variation of this algorithm can be implemented in the MPC model with a constant number of rounds, achieves a $2-1/8$-approximation. 
\begin{theorem}
For any $\epsilon>0$, there is an algorithm that w.h.p., achieves a $2-2/13+\epsilon$-approximation for 
Correlation Clustering in time $O(2^{\epsilon^{-O(1)}}n\log n)$, or in space $O(2^{\epsilon^{-O(1)}}n\log n)$ in the streaming setting.

In addition, for any $\delta = \Omega(1)$, there is an algorithm that w.h.p., achieves a $2-1/8+\epsilon$-approximation and can be implemented to terminate in 
$(\epsilon\delta)^{-O(1)}$ rounds in the MPC model with $n^{\delta}$ memory per machine and
total memory $O(\epsilon^{-O(1)} m \log n)$.
\end{theorem}
Fixing $\epsilon=0.0008$, we
get an approximation factor of
$2-2/13+\epsilon<1.847$ and 
$2^{\epsilon^{-O(1)}}=O(1)$ though with a very large hidden constant.
Our approach thus opens up the possibility of achieving a better than 2 
approximation in other models such as dynamic or CONGEST,
and it is an interesting direction to see whether this can be achieved.

\subsection{Techniques}
The local search framework we consider maintains a clustering. At 
each iteration, the algorithm tries to swap in a cluster (i.e.: a set of vertices
of arbitrary size), creating it as a new cluster and removing its elements
from the clusters they belonged to. The swap
is improving if the resulting clustering
has a smaller cost, and we keep doing improving swaps until no more are found.

Of course, as described above the algorithm would have a running time exponential in the
input size since it for potential swaps needs to 
consider all subsets of vertices of the input but let's
first put aside the running time consideration and focus on
the approximation guarantee. We first discuss that the algorithm achieves a 2-approximation;
this serves as a backbone for our further improvements. We further discuss how we can
implement the above variant of the above algorithm in polynomial time while losing only
a $(1+\epsilon)$ factor in the approximation guarantee. 

\paragraph{Achieving a 2-approximation using a simple (exponential time) local search}
The analysis of the above algorithm simply consists of the following steps: Let $S$ be
the solution output by the algorithm and $\opt$ the optimum solution. For each cluster
$C$ in $\opt$ we will consider the solution $S_C$ which is created from $S$ by inserting $C$ as a cluster
(by creating cluster $C$ and removing its elements from the clusters they belong to in $S$).

By local optimality, we know that the cost of solution $S$ is no larger than the cost of the
solution $S_C$. The difference in cost between the two solutions is only coming from edges
that are adjacent to at least one vertex in $C$. More concretely, denote by $E^+_S$ the set
of edges cut by $S$ and by $E^-_S$ the set of non-edges that are not cut by $S$, i.e.:
the cost of $S$ is equal to $|E^+_S|+|E^-_S|$. Then the cost of solution $S_C$
minus the cost of solution $S$ is positive and equal to
\begin{align*}
0 \le 
|\{(u,v) \mid (u,v) \in E \setminus E^+_S \text{ and either } u \text{ or } v \in C\}|+|\{(u,v) \mid (u,v) \notin E \cup E^-_S
\text{ and both } u,v \in C\}|
\\- |\{(u,v) \mid (u,v) \in E^+_S \text{ and both } u,v \in C\}|-
|\{(u,v) \mid (u,v) \in E^-_S \text{ and either } u\text{ or }v \in C\}| 
\end{align*}

This statement holds for any cluster $C$ of the optimum solution and
so summing up the above inequality over all these clusters shows 
that on the left-hand side, the cost of the local search solution is 
no more than twice the cost of the optimum solution. 
This hence provides a simple (exponential time) 2 approximation algorithm for Correlation Clustering.

Equipped with the above analysis, our work then moves in two complementary directions:
(1) Improving over the approximation bound of 2 with certain flips; and 
(2) showing how to implement the above algorithm while losing at most
a $(1+\epsilon)$ factor in the overall guarantee.

\paragraph{Flips between local searches}
Let us first focus on (1). When considering the above analysis, one 
can note that the factor of 2 arises when summing over all 
clusters and so each edge $(u,v) \in E \setminus E^+_S$ where $u \in C$
and $v \in C' \neq C$ in the optimum solution. Indeed, in this case,
the edge $(u,v)$ appears twice in the sum: once when upper bounding the
cost of the vertices in $C$ and once when upper bounding the cost of the vertices
in $C'$. On the other hand, the ratio coming from either the non-edges
or the pairs $(u,v) \in E^+_S \cup E^-_S$ that are also paid by $\opt$ is in fact 1.

It follows that the solution $S$ obtained by local search is only close to a 2 
approximation if a, say, .99 fraction of the cost of the optimum solution is 
due to the edges (and not the non-edges) cut by the optimum solution and
that are not in $E^+_S$. 

Hence, for the ratio of 2 to be tight, the solution $S$, the optimum
solution and the underlying graph must be very peculiar and satisfy
the following:
\begin{itemize}
    \item .99 fraction of the cost of the optimum solution is due to edges (and not non-edges)
    \item $S$ and the optimum solution must have at most .01 fraction
    of their cost shared (i.e.: intersection of the edges and non-edges
    paid by the $\opt$ and $S$ must be at most a .01 fraction of their overall cost). 
\end{itemize}
Under our assumptions, we see that
$\opt$ and $S$ share very few cut edges, so to get closer to $\opt$, it would make sense to \emph{flip} the cut edges of $S$; seeking a solution where they are no longer cut. We do this by increasing the cost of cutting these edges.

The above may seem a bit naive, but it turns out to yield an approximation factor substantially below 2. More precisely, after having found 
the first local search solution $S$, we double the cost of cutting the weight of the cut edges of $S$. This is called the \emph{flip}. Then we run a new local search on this
modified input; yielding a new local optimum $S'$.
Finally, we return the one of $S$ and $S'$ that minimizes the original cost. Magically, it turns out that this local-search-flip-local-search always provides a $15/8$-approximation.
We provide a description of this algorithm and its approximation ratio as a warm-up
in Section~\ref{sec:warmup}.

\paragraph{Iterated-flipping Local Search}
We then go one step beyond and analyze the bad instances for the solution
output after we perform the flip (i.e.: the weighting of the edges cut by the initial
solution $S$).
Our first new ingredient is to start iterating the flip process with different
adaptive weights on the edges cut by the solutions output by the local search algorithm on
the different weighted graphs (the motivation is again that the bad case in these scenarios
is when the above two assumptions are nearly satisfied).
In this case, we use the sequence of solutions found by applying the local search to further
and further refine our weighting scheme. 
The second key ingredient is a procedure that pivots over three clusterings to output a new
clustering that combines them. The pivot procedure iteratively creates a new clustering by
taking the largest set of vertices that are together in all three clusterings,
and creating a new cluster consisting of vertices that are together with 
the chosen ones in at least two out of three clusterings, 
and repeating on the remaining vertices.
We then build a sequence of solutions that iteratively modify the weights of the edges of
the graph and pivot on the three last solutions created. We show that the best solution among
the ones created achieves our final bound of $2-2/13$.

\paragraph{Efficient implementations}
In the previous paragraph, we have worked with the idealistic local 
search algorithm that can swap in and out clusters of arbitrary sizes.

A striking phenomenon here is that it is impossible to approximate the
cost of the optimum solution up to any $o(n)$ factor on graphs with $n$ vertices
in the sublinear memory regime; hence all previous algorithms on this model
output a solution without estimating its cost. In light of this, it may seem
hopeless to use a local-search-based approach that requires estimating the
cost of swapping in and out clusters to improve the cost of the current solution.
Yet, we show that thanks to the coarse preclustering, sampling techniques are enough to estimate the cost of potential
swaps and implement our approach.

Another crucial feature of our local-search-based approach is that it
can also be implemented in a distributed environment and allow to 
perform local search iterations in parallel.

We now discuss the main ideas that led to achieving a linear running time, and derive algorithms for 
the sublinear time
and massively-parallel computation models. Let's focus on implementations for the simple
local search presented at the beginning of this section. 
To implement our approach we need to show how to implement the following key primitive:
Given a clustering $S$, is there a cluster $C$ such that the solution $S_C$ has a better
cost than $S$. Assume first that we are in the case where the optimum solution has cost
at least $\epsilon$ times the total number of edges of the graph.
In this case, we can tolerate an additive cost proportional to $\epsilon^2$ times the sum
of the degrees of the vertices in $C$ and our analysis would be unchanged up to an additive
term of $2\epsilon$ times the cost of $\opt$.

In this context, we make use of the fact that the clusters of any optimum solution are
dense (i.e.: each cluster $C$ of $\opt$ contains at least $|C|(|C|-1)/2 -1$ edges since
otherwise, the vertices could be placed into a singleton and the cost would improve).
This means that we can pick uniformly at random a set of, say, $\text{poly}(1/\epsilon)$
vertices from the cluster and this provides an accurate estimate of the fraction of
neighbors in the cluster for any other vertex $v$ in the graph, up to an additive
$\epsilon |C|$ term. Moreover, we show that there exists a near-optimal solution $\opt^*$
where each vertex in $C$ has degree proportional to $|C|$. Our algorithm then aims
at sampling a set of $\text{poly}(1/\epsilon)$ vertices from $C$, and places
these vertices in a tentative cluster $C'$. It then 
iterates over a set of \emph{candidate} vertices that
could tentatively be included in $C'$ (and that contains
the vertices in $C$).
We would then like to show that the total
amount of mistakes made by the algorithm is proportional to the sum of the degrees
of the vertices in $C$. To achieve this, we require two more ingredients: (1) the candidate vertices have degree proportional to $|C|$ and
their number is proportional to $|C|$; and (2) that these greedy
steps are consistent with the uniform sample we have selected, namely that the uniform
sample we are using to make our decisions is also a uniform sample of $C'$ (since otherwise
the decisions made are not relevant w.r.t. the final cluster produced).

To both achieve (1) and ensure that we can tolerate an error of $\epsilon$ times the sum of
the edges in the instance, we use the Preclustering algorithm recently introduced in~\cite{CLLN23+}
which merges vertices that they prove can be clustered together in a near-optimal solution. 

To achieve (2), we use an iterated approach where the cluster is constructed in batches
of size $\epsilon^{q}|C|$ for some constant $q$ and the sampling is ``updated'' after each batch.
Decisions made for the $i$th batch are based on a sample
obtained by sampling the set of vertices that have already joined $C'$ and the set of
vertices that we would like to join $C'$. The mistakes for the $i$th batched are thus proportional
to $\epsilon^{2q}|C|$ and since the total number of batches is $O(1/\epsilon^{q})$ (thanks to (1)) the
final mistake obtained is $\epsilon^{q}|C|$. 

Finally, we then show that the algorithm can be implemented in linear time and in fact sublinear time
and MPC models. The key observation here is that when looking for a new cluster $C$ to swap in, one
can focus on a vertex $v$ of $C$ and look at the set of vertices with similar degree to $v$ and that
are adjacent to $v$ or to a neighbor of $v$ with degree similar to $v$. The size of this set is proportional
to the degree of $v$ and so proportional to the cluster $C$. Thus, in time proportional to the degree of $v$
one can run the above procedure to build the cluster $C$.

\subsection{Further Related Work}
If the number of clusters is a hard constraint and a constant, then
a polynomial-time approximation
scheme  exists~\cite{giotis2006correlation,karpinski2009linear}
(the running time depends exponentially on the number of clusters).
As hinted at in the introduction, there is a large body of on 
Correlation Clustering in other computation models: there exists
online algorithms~\cite{mathieu2010online,NEURIPS2021_250dd568,CLMP22},
dynamic algorithms~\cite{BehnezhadDHSS-FOCS19}, distributed or
parallel algorithms~\cite{chierichetti2014correlation,pmlr-v37-ahn15,CLMNP21,DBLP:conf/nips/PanPORRJ15,DBLP:conf/wdag/CambusCMU21,DBLP:conf/icml/Veldt22,DBLP:conf/www/VeldtGW18,DBLP:conf/focs/BehnezhadCMT22,CaoHS-SODA24},
differentially-private algorithms~\cite{bun2021differentially,Daogao2022,CFLMNPT22}, 
sublinear time algorithms~\cite{DBLP:conf/innovations/Assadi022}, 
and streaming algorithms~\cite{DBLP:conf/soda/BehnezhadCMT23,DBLP:conf/innovations/Assadi022,CambusKLPU-SODA24,CLMNP21}. Note that recently,~\cite{DBLP:conf/focs/Cohen-Addad0KPT21,CFLM22} have 
establish a connection between Correlation Clustering and metric and
ultrametric embeddings.

In the case where the edges of the graph are arbitrarily weighted, and 
so the cost induced by violating an edge is its weight, the problem 
is as hard to approximate as the Multicut problem (and in fact ``approximation equivalent'') and an $O(\log n)$-approximation is
known~\cite{demaine2006correlation}. The work of 
Chawla, Krauthgamer, Kumar, Rabani, and Sivakumar thus implies that
there is no polynomial-time constant factor approximation algorithm 
for weighted Correlation Clustering assuming the Unique Game Conjecture~\cite{CKKRS06}.

If we flip the objective and aim at  maximizing
the number of non-violated edges, a PTAS has been  given by Bansal, Blum and Chawla~\cite{BBC04}, and a .77-approximation if the edges of the graph are weighted was given by Charikar, Guruswami and Wirth~\cite{CGW05} and
Swamy~\cite{swamy2004correlation}.

\section{Preliminaries}

\paragraph{Graph notation.}
We will work with simple undirected edges. For a graph $G$,
by $V(G)$ and $E(G)$ we denote its vertex and edge set, respectively.
If the graph is clear from the context, we denote $E^+ = E(G)$
and $E^- = \binom{V(G)}{2} \setminus E(G)$, that is, the set of edges and nonedges of $G$.

Let $d(v)$ the degree of vertex $v$, and for any subset $C$ of vertices
let $d(v,C)$ be the number of neighbors of $v$ in $C$.

Some of the graphs will be accompanied by a weight function $w : \binom{V(G)}{2} \to \mathbb{R}_+$.
The unweighted setting corresponds to $w \equiv 1$.
For any set $D \subseteq \binom{V(G)}{2}$, $w(D)$ denotes the total weight of pairs in $D$.
When the graph is unweighted, $w(D)=|D|$.

\paragraph{Correlation clustering.}
A \emph{clustering scheme} in $G$ is a partition of $V(G)$. 
Every set in a clustering scheme is called a \emph{cluster}. 
For a clustering scheme $\calC$, let $\Int(\calC)$ be the set of pairs of vertices that are in the same
cluster of $\calC$ and $\Ext(\calC)$ be the set of edges across different clusters.
$\cost_w^+(\calC)=w(E^+\cap\Ext(\calC))$ denotes the ``$+$'' cost of $\calC$.
$\cost_w^-(\calC)=w(E^-\cap\Int(\calC))$ denotes the ``$-$'' cost of $\calC$.
$\cost_w(\calC)=\cost_w^+(\calC)+\cost_w^-(\calC)$ denotes the total cost of $\calC$.
We drop the subscript if the weight function $w$ is clear from the context.

The input to \textsc{Correlation Clustering} is a simple undirected graph $G$, 
with additionally a weight function $w$ in the weighted variant. 
We ask for a clustering scheme of the minimum possible cost. 

We will also use the notation $\Int$ and $\Ext$ for individual clusters $C \subseteq V(G)$.
Then, $\Int(C) = \binom{C}{2}$ and $\Ext(C) = \{uv \in \binom{V(G)}{2}~|~|\{u,v\} \cap C| = 1\}$. 

\paragraph{Preclustering.}
Let $G$ be an unweighted instance to \textsc{Correlation Clustering} 
and let $\opt$ be a clustering scheme of minimum cost.
Note that every $v \in V(G)$ is in a cluster of size at most $2d(v)+1$, as otherwise
it would be cheaper to cut $v$ out into a single-vertex cluster. 
However, $v$ can be put in $\opt$ in a cluster much smaller than $d(v)$.

We can ensure that $v$ is put in a cluster of size comparable with $d(v)$
allowing a small slack in the quality of the solution.
Let $\varepsilon > 0$ be an accuracy parameter. Note that
if the cluster of $v$ in $\opt$ is of size smaller than $\frac{\varepsilon}{1+\varepsilon} d(v)$,
then one can cut out $v$ into a single-vertex cluster and charge the cost to at least
$\frac{1}{1+\varepsilon}d(v)$ already cut edges incident with $v$. 
That is, there exists a $(1+\varepsilon)$-approximate solution with the following property:
every $v\in V(G)$ is either in a single-vertex cluster or in a cluster of size
between $\frac{\varepsilon}{1+\varepsilon} d(v)$ and $2d(v)+1$. 
For fixed $\varepsilon > 0$, this is of order $d(v)$.

Preprocessing introduced in~\cite{CLLN23+} takes the above approach further and
also identifies some pairs of vertices that can be safely assumed to be together or separate
in a $(1+\varepsilon)$-approximate clustering scheme. 
The crux lies in the following guarantee: the number of vertex pairs of ``unknown'' status
is comparable with the cost of the optimum solution.

We now formally state the outcome of this preprocessing.
\begin{definition}\label{definition:prepro}
For an unweighted \textsc{Correlation Clustering} instance $G$, a {\em preclustered instance} is a pair $(\calK, E_{\adm})$, where 
$\calK$ is a family of disjoint subsets of $V(G)$ (not necessarily a partition),
and $E_{\adm} \subseteq \binom{V(G)}{2}$ is a set of pairs called \emph{admissible pairs}
such that for every $uv \in E_\adm$, at least one of $u$ and $v$ is not in $\bigcup \calK$.

Each set $K \in \calK$ is called an {\em atom}. We use $V_\calK:= \bigcup \calK$ to denote the set of all vertices in atoms.
A pair $uv \in \binom{V(G)}{2}$ with $u$ and $v$ in the same $K \in \calK$ is called an \emph{atomic pair}.
A pair $(u, v)$ between two vertices $u,v$ in the same $K \in \calK$ is called an {\em atomic edge}.
A pair that is neither an atomic pair nor an admissible pair is called a \emph{non-admissible pair}. 
\end{definition}

For a vertex $u$, let $d^{\adm}(u)$ denote the number of $v \in V(G)$ such that $uv$ is an admissible pair.
Note that $\sum_{u\in V(G)}d^{adm}(u)=2|E_{\adm}|$.

\begin{definition}\label{def:precluster}
Let $G$ be an unweighted \textsc{Correlation Clustering} instance, let $(\calK, E_{\adm})$ be a preclustered instance for $G$, and let $\varepsilon > 0$ be an accuracy parameter.
We say that $(\calK, E_{\adm})$ is 
    a \emph{$\varepsilon$-good} preclustered instance 
    if the following holds:
    \begin{enumerate}
         \item for every $v \in V(G)$, we have $d^\adm(v) \le 2\varepsilon^{-3} \cdot d(v)$, and 
        \item for every $uv \in E_\adm$, we have $d(u) \le 2 \varepsilon^{-1} \cdot d(v)$.
        \item for every atom $K \in \calK$, for every vertex $v \in K$, we have that $v$ is adjacent
        to at least a $(1-O(\epsilon))$ fraction
        of the vertices in $K$ and has at most $O(\epsilon |K|)$ neighbors not in $K$.
    \end{enumerate}
\end{definition}
Therefore, in a preclustered instance, the set $\binom{V(G)}{2}$ is partitioned into atomic, admissible, and non-admissible pairs.
By the definition of $E_\adm$, a pair $uv$ between two different atoms is non-admissible. 

\begin{definition}[Good clusters and good clustering schemes]
Let $G$ be an unweighted \textsc{Correlation Clustering} instance, let $(\calK, E_{\adm})$ be a preclustered instance for $G$, and let $\varepsilon > 0$ be an accuracy parameter.
Assume that $(\calK,E_\adm)$ is $\varepsilon$-good.
For a real $\delta > 0$, a set $C \subseteq V(G)$ is called a \emph{$(\varepsilon,\delta)$-good} cluster with respect to $(\calK, E_{\adm})$ if 
\begin{itemize}
    \item for any $u$ in $C$, for any 
    $v$, if $uv$ is an atomic pair, then $v$ is also in $C$.
    \item for any distinct $u,v$ in $C$, we have
    that $uv$ is not a non-admissible pair.
    \item for any vertex $v$ in $C$, if $|C| > 1$, then $|C| \ge \delta d(v)$.
\end{itemize}

Moreover, a clustering scheme $\calC$ is called {\em $(\varepsilon,\delta)$-good} with respect to $(\calK, E_{\adm})$ if all clusters $C \in \calC$ are $(\varepsilon,\delta)$-good. We will also say a cluster/clustering scheme is $\varepsilon$-good if it is $(\varepsilon,\varepsilon)$-good.
\end{definition}

That is, a good cluster can not break an atom, or join a non-admissible pair.
As a result, two atoms can not be in the same cluster as the pairs between them are non-admissible. 

We consider the following algorithm to compute a preclustering due to~\cite{CLLN23+}
\begin{tcolorbox}[beforeafter skip=10pt]
    \textbf{Preclustering($G,\varepsilon$)}
    \begin{itemize}[itemsep=0pt, topsep=4pt]
    \item Run Algorithm 1 of ~\cite{CLMNP21}. Each non-singleton cluster is an atom. let $\calK$ be the set of atoms produced.
    \item A pair $uv$ is \emph{degree-similar} if $\varepsilon d(v) \le d(u) \le d(v) / \varepsilon$.
A pair $uv$ is \emph{admissible} if (1) either $u$ or $v$ belongs to $V(G) \setminus{V_\calK}$, and 
(2) it is degree-similar, and (3) the number
of common neighbors that are degree-similar to both $u$ and $v$ is at least 
$\varepsilon \cdot \text{min}\{d(u), d(v)\}$.
We let $E_\adm$ be the set of all admissible pairs in $\binom{V(G)}{2}$.  
\item Output the pair $(\calK, E_{\adm})$.
    \end{itemize}
\end{tcolorbox}

The following theorem is in fact a restatement from~\cite{CLLN23+}: The only addition 
from the Theorem 4 of~\cite{CLLN23+} (arXiv version)
is that we require the instance to be $\varepsilon$-good: Namely that the 3 bullets of Definition~\ref{def:precluster} hold.
The last two bullets are shown to be true in the proof 
of Theorem 4 of~\cite{CLLN23+} (arXiv version) while
the last bullet follows from the description of 
Algorithm 1 of~\cite{CLMNP21}.

\begin{theorem}[Preclustering~\cite{CLLN23+,CLMNP21}]
\label{thm:preprocessed-wrapper}
Let $\varepsilon > 0$ be a sufficiently small accuracy parameter.
Then, given an unweighted \textsc{Correlation Clustering} instance $G$, the algorithm Preclustering($G,\varepsilon$)
  produces a $\varepsilon$-good preclustered instance $(\calK,E_{\adm})$ that 
  admits an (unknown to the algorithm) $\varepsilon$-good clustering scheme $\calC^*_{(\calK, E_{\adm})}$ with the following properties:
        \begin{itemize}
            \item $\cost(\calC^*_{(\calK, E_{\adm})})$ is at most $(1+\varepsilon)$ times the minimum cost of a clustering scheme in $G$;
            \item $\cost(\calC^*_{(\calK, E_{\adm})})$ is at least $(\varepsilon^{12}/2) \cdot |E_{\adm}|$.
        \end{itemize}
\end{theorem}

 \subsection{Sublinear, streaming and MPC models}
We further define the sublinear, streaming and MPC models. Let $G=(V,E)$ be a graph.
\begin{definition}
In the \emph{sublinear} model, the algorithm can query the following information in $O(1)$ time:
\begin{itemize}
    \item Degree queries: What is the degree of vertex $v$?
    \item Neighbor queries: What is the $i$-th neighbor of $v \in V$ for $i$ less or equal to the degree of $v$?
\end{itemize}
\end{definition}

\begin{definition}
In the \emph{streaming} model, the algorithm knows the vertices of the graph in advance, and the edges arrive in a stream. The algorithm only has $O(n\operatorname{polylog}(n))$ space and can only read the stream once.
\end{definition}

In the MPC model, the set of edges is distributed to a set of machines. Computation then proceeds in
synchronous rounds. In each round, each machine first receives
messages from the other machines (if any), then performs some 
computation based on this information and its own internal allocated 
memory,  and terminates the round by sending messages to other machines
that will be received by the machine at the start of the next round. Each message has size $O(1)$ words. Each machine has limited local memory, which restricts the total number of messages it can receive or send in a round.

For the Correlation Clustering problem, the algorithm's output at the
end of the final computation round is that each machine is required
to store in its local memory the IDs of the cluster of the vertices
that it initially had edges adjacent to. 

We consider the strictly sublinear MPC regime where each machine 
has local memory $O(n^{\delta})$, where $\delta > 0$ is a constant that can be made arbitrarily small, and $\tilde{O}(m)$ total space.

The theorem below is immediate from Theorem 1 in \cite{CLMNP21}.
\begin{theorem}[Preclustering in the MPC and sublinear time model -- \cite{CLMNP21} Theorem 1]
\label{thm:sublinear-mpc-preclustering}
For any constant $\delta > 0$,  Preclustering$(G,\epsilon)$ can be 
implemented in the MPC model in $O(1)$ rounds. 
Letting $n = |V|$, this algorithm succeeds with probability at least 
$1 - 1/n$ and requires $O(n^\delta)$ memory per
machine. Moreover, the algorithm uses a total memory of
$O(|E|  \log n)$.
\end{theorem}

The following sampling techniques will be used in our implementations in the sublinear and streaming models.

\begin{lemma}[\cite{DBLP:conf/innovations/Assadi022}]
\label{lem:sample-neighbor}
    In $O(n\log n)$ time (in the sublinear model) or in $O(n\log n)$ space (in the streaming model), we can sample $\Theta(\log n)$ neighbors (with repetition) of each vertex $v$.
\end{lemma}

\begin{lemma}[\cite{DBLP:conf/innovations/Assadi022}]
\label{lem:sample-vertex}
    In $O(n\log n)$ time (in the sublinear model) or in $O(n\log n)$ space with high probability (in the streaming model), we can sample each vertex $v$ with probability $\Theta(\log n/d(v))$ and store all edges incident to sampled vertices.
\end{lemma}

We remark that in the streaming model, we can also get several independent realizations of the above sampling schemes in one pass which only increases the space complexity.

The next theorem follows the decomposition theorem and the framework 
of Assadi and Wang~\cite{DBLP:conf/innovations/Assadi022,CLMNP21,CLLN23+}.

\begin{theorem}[\cite{DBLP:conf/innovations/Assadi022,CLMNP21,CLLN23+}]
\label{thm:admissible-query}
In $O(n \log^2 n)$ time (in the sublinear model) or in $O(n\log n)$ space (in the streaming model), one can compute a partition $\calK$ of the vertices and a data structure such that 
there exists an $\epsilon$-good preclustering $(\calK,E_{\adm})$ satisfying the conditions
of Theorem~\ref{thm:preprocessed-wrapper} with 
\begin{itemize}
    \item $\cost(\calC^*_{(\calK, E_{\adm})})$ is at most $(1+\varepsilon)$ times the minimum cost of a clustering scheme in $G$;
\item $\cost(\calC^*_{(\calK, E_{\adm})}) > (\varepsilon^{12}/2^{13}) \cdot |E_{\adm}|$ 
\end{itemize} 
and such that with success probability at least $1-1/n^2$, 
for any pair of vertices, the data structure can answer in $O(\log n)$ time whether the edge is in $E_{\adm}$ or not, 
and for any vertex $v$, the data structure can list all vertices admissible to $v$ in $O(d(v)\log^2 n)$ time.
\end{theorem}
The proof of the above statement follows from the following observation:
The set $\calK$ is partitioned into disjoint dense vertex set computed
by the framework of Assadi and Wang. We initialize one realization of Lemma \ref{lem:sample-vertex}, which samples each vertex $v$ with probability $40\varepsilon^{-2}\log n/d(v)$. With high probability, each vertex $v$ will only have $O(\log n)$ degree-similar sampled neighbors. Then, when the data structure
is queried with a pair $u,v$ it can implement the second step of 
the preclustering by (1) verifying that none of $u,v$ belong to 
$\calK$ in $O(1)$ time, (2) verifying that $u,v$ are degree similar
with two degree queries, and (3) looking at the degree-similar sampled
neighbors of $u$ and $v$ in total time $O(\log n)$ and 
verifying that the intersection is at least $\varepsilon \min\{d(u),d(v)\}/2$
of the Preclustering algorithm. 
An immediate application of Chernoff bound implies that the probability of misreporting that an edge is not admissible is at most
$1/n^4$. Then by a union bound, with probability $1-1/n^2$ it can correctly answer all queries. On the other hand, the number of edges reported is the 
same as in Theorem~\ref{thm:preprocessed-wrapper} applied to an
$\varepsilon$ twice smaller.
Then, when we want to list all vertices admissible to $v$, we look at every degree-similar sampled neighbor $u$ of $v$. For each neighbor of $u$, we query whether it is admissible to $v$. Given that the data structure can correctly answer all queries, we must find all candidates in this way, and the total number of queries we make is $O(d(v)\log n)$.

\section{Local Search with Flips}
\label{sec:LS}
In this section, we present and analyse a few local search based approaches to 
\textsc{Correlation Clustering}. 

First, we present a simple warm-up $(2-\frac{1}{8})$-approximation. 
Here, we will not discuss the implementation of a local search step, but only properties
of a local optimum. 
Then, we present a more involved $(2-\frac{2}{13}+\varepsilon)$-approximation, where we in the analysis also care to leave placeholders for our implementation of a local search step.
This implementation will be via sampling, and can only guarantee being in 
an ``$(1+\varepsilon)$-approximate'' local optimum.

Let $\calC$ be a clustering scheme in a graph $G$ with weight function $w$ and let 
$C \subseteq V(G)$. 
By $\calC + C$ we denote the clustering scheme obtained from $\calC$ by first removing
the vertices of $C$ from all clusters and then creating a new cluster $C$. 

\subsection{Warm-up: $(2-\frac{1}{8})$-approximation}
\label{sec:warmup}
As a warm-up, to show our techniques, we present an analysis of a local optimum of a simple local search. 
\begin{definition}
A clustering scheme $\ls$ is a \emph{local optimum} if
for every $C \subseteq V(G)$ we have $\cost(\ls + C) \geq \cost(\ls)$.     
\end{definition}

We will analyse the following algorithm.
\begin{tcolorbox}[beforeafter skip=10pt]\label{alg:local-search-flip}
    \textbf{Local-Search-Flip-Local-Search}
    \begin{itemize}[itemsep=0pt, topsep=4pt]
        \item Let $\ls1$ be a local optimum of $(G,w)$.
        \item Double the weight of edges in $E^+\cap\Ext(\ls1)$ to get a new weight function $w'$.
        \item Let $\ls2$ be a local optimum of $(G,w')$.
        \item Output the best of $\ls1$ and $\ls2$ with respect to the original cost function.
    \end{itemize}
\end{tcolorbox}

The remainder of this section is devoted to the proof of the following theorem.
\begin{theorem}
\label{thm:two-ls}
   The Local-Search-Flip-Local-Search algorithm applied to an instance with uniform weight ($w \equiv 1$) outputs a solution of cost within a
   $(2-\frac{1}{8})$ ratio of the optimum solution.
\end{theorem}

We start the proof with a few observations about local optimum.
\begin{lemma}\label{lem:prop-ls}
Let $\ls$ be a local optimum and let $\calC$ be a clustering scheme. Then, the following holds.
\begin{align}
\cost^-(\calC)+2\cost^+(\calC) &\ge w(E^+\cap\Ext(\ls)\cap\Int(\calC))+2w(E^+\cap\Ext(\ls)\cap\Ext(\calC))\nonumber\\
    &\quad +w(E^-\cap\Int(\ls)\cap\Int(\calC))+2w(E^-\cap\Int(\ls)\cap\Ext(\calC)),\label{eq:prop-ls}\\
\cost(\ls) &\le 2\cost(\calC)-\cost^-(\calC)-w(E^-\cap\Int(\ls)\cap\Ext(\calC))\nonumber\\
  &\quad -w(E^+\cap\Ext(\ls)\cap\Ext(\calC)),\label{eq:prop-ls-2}\\
\cost(\ls) &\le 2\cost(\calC)-\cost^-(\calC)-w(E^+\cap\Ext(\ls)\cap\Ext(\calC)),\label{eq:ls-cost-1}\\
\cost(\ls) &\le 2\cost(\calC)-\cost^-(\ls)-w(E^+\cap\Ext(\ls)\cap\Ext(\calC)).\label{eq:ls-cost-2}
\end{align}
\end{lemma}
\begin{proof}
Since $\ls$ is a local optimum, for every cluster $C$ of $\calC$
it holds that $\cost(\ls + C) \geq \cost(\ls)$. That is,
\begin{align*}
w(E^- \cap \Int(C)) + w(E^+ \cap \Ext(C)) &\geq w(E^+ \cap \Int(C) \cap \Ext(\ls)) + w(E^+ \cap \Ext(C) \cap \Ext(\ls)) \\
    &\quad + w(E^- \cap \Int(C) \cap \Int(\ls)) + w(E^- \cap \Ext(C) \cap \Int(\ls)).
\end{align*}
Summing the above over every cluster $C$ of $\calC$ proves~\eqref{eq:prop-ls}, as every pair
of $\Ext(\calC)$ will appear twice and every pair of $\Int(\calC)$ will appear once
in the summands.

Inequality~\eqref{eq:prop-ls-2} follows from~\eqref{eq:prop-ls} by simple algebraic manipulation,
  noting that $\cost(\calC) = \cost^+(\calC) + \cost^-(\calC)$ and
  $\cost(\ls) = w(E^+\cap\Ext(\ls)\cap\Int(\calC))+w(E^+\cap\Ext(\ls)\cap\Ext(\calC))
  + +w(E^-\cap\Int(\ls)\cap\Int(\calC))+w(E^-\cap\Int(\ls)\cap\Ext(\calC))$.
 
Inequality~\eqref{eq:ls-cost-1} is just~\eqref{eq:prop-ls-2} with
the nonnegative term $w(E^- \cap \Int(\ls) \cap \Ext(\calC))$ dropped.

For~\eqref{eq:ls-cost-2}, we use the following immediate estimate:
\begin{equation}\label{eq:ls-extra}
 w(E^-\cap\Int(\ls)\cap\Int(\calC)) \le \cost^-(\calC).
\end{equation}
Inequality~\eqref{eq:ls-cost-2} follows from~\eqref{eq:prop-ls-2} by first 
adding~\eqref{eq:ls-extra} to the sides and then observing
that $\cost^-(\ls) = w(E^- \cap \Int(\ls) \cap \Int(\calC)) + w(E^- \cap \Int(\ls) \cap \Ext(\calC))$.
\end{proof}

Let $\delta = \frac{1}{8}$. In the proof, it is instructive to think about $\delta$
as a small constant; only in the end we will observe that setting $\delta = \frac{1}{8}$ gives
the best result.

Let $\opt$ be an optimum solution. We say that a clustering scheme $\calC$ is
\emph{$\alpha$-competitive} if $\cost(\calC) \leq \alpha \cost(\opt)$.

The structure of the proof of Theorem~\ref{thm:two-ls} is as follows:
assuming that both $\ls1$ and $\ls2$ are not $(2-\delta)$-competitive,
we construct a solution which is better than $8\delta$-competitive.
This finishes the proof, as we cannot have a solution better than the optimum one. 

To execute this line, we analyse in detail costs of various sets of (non)edges of $G$
using inequalities of Lemma~\ref{lem:prop-ls}.
\begin{lemma}
\label{lem:ls1}
If $\ls1$ is not $(2-\delta)$-competitive, then the following holds.
\begin{align*}
\cost^-(\opt)+w(E^+\cap\Ext(\ls1)\cap\Ext(\opt)) &< \delta\cost(\opt),\\
\cost^-(\ls1)+w(E^+\cap\Ext(\ls1)\cap\Ext(\opt)) &< \delta\cost(\opt).
\end{align*}
\end{lemma}

\begin{proof}
Immediately from~\eqref{eq:ls-cost-1} and~\eqref{eq:ls-cost-2} for $\ls1$ and $\opt$.
\end{proof}

Here we use $\cost'$ to denote the cost on the instance $(G,w')$.
By the definition of $w'$, for any clustering scheme $\calC$,
\begin{equation}
\label{eqn:new-cost}
    \cost'(\calC) = \cost(\calC)+w(E^+\cap\Ext(\ls1)\cap\Ext(\calC)).
\end{equation}

We remark that $\opt$ still denotes the optimal solution of $(G,w)$, and the ``$-$'' cost for any clustering scheme remains the same in $(G,w')$.

\begin{lemma}
\label{lem:ls2}
    If $\ls2$ is not $(2-\delta)$-competitive, then 
\begin{align*}
& \cost^-(\ls2)+w(E^+\cap\Ext(\ls2)\cap\Ext(\opt))+w(E^+\cap\Ext(\ls1)\cap\Ext(\ls2)) \\
  &\qquad < \delta\cost(\opt) + 2w(E^+\cap\Ext(\ls1)\cap\Ext(\opt)).
\end{align*}
\end{lemma}

\begin{proof}
    By Equation \eqref{eqn:new-cost}, we have:
    \begin{align*}
    \cost'(\opt) &= \cost(\opt)+w(E^+\cap\Ext(\ls1)\cap\Ext(\opt)),\mathrm{\ and}\\
    \cost'(\ls2) &= \cost(\ls2)+w(E^+\cap\Ext(\ls1)\cap\Ext(\ls2)).
    \end{align*}
    
By applying~\eqref{eq:ls-cost-2} to $\ls2$ and $\opt$ in $(G,w')$, we obtain:
    \begin{align*}
    \cost'(\ls2) &\le 2\cost'(\opt)-\cost^-(\ls2)-w'(E^+\cap\Ext(\ls2)\cap\Ext(\opt))\\
      &\le 2\cost'(\opt)-\cost^-(\ls2)-w(E^+\cap\Ext(\ls2)\cap\Ext(\opt)).
    \end{align*}
    Here, the last inequality follows from the fact that all weights
    of $w$ are nonnegative and after the flip $w'$ is $w$ with some weights doubled.

    Combining them, we get:
    \begin{align*}
        \cost(\ls2) \le & \ 
        2(\cost(\opt)+w(E^+\cap\Ext(\ls1)\cap\Ext(\opt)|)) - \cost^-(\ls2) \\
        & -w(E^+\cap\Ext(\ls2)\cap\Ext(\opt))-w(E^+\cap\Ext(\ls1)\cap\Ext(\ls2)).
    \end{align*}

    The lemma follows.
\end{proof}

We combine the estimates of Lemmata~\ref{lem:ls1} and~\ref{lem:ls2} into the following.
\begin{lemma}
\label{lem:special-cost}
    If both $\ls1$ and $\ls2$ are not $(2-\delta)$-competitive, then
    \begin{align*}
    &\cost^-(\opt)+\cost^-(\ls1)+\cost^-(\ls2)+w(E^+\cap\Ext(\ls1)\cap\Ext(\opt))\\
    &+w(E^+\cap\Ext(\ls2)\cap\Ext(\opt))+w(E^+\cap\Ext(\ls1)\cap\Ext(\ls2))\\
    &\qquad < 4\delta\cost(\opt).
    \end{align*}
\end{lemma}

\begin{proof}
    By Lemma~\ref{lem:ls2}, 
    the left-hand-side is at most
    \[\cost^-(\opt)+\cost^-(\ls1)+3w(E^+\cap\Ext(\ls1)\cap\Ext(\opt))+\delta\cost(\opt),\]
    which, by Lemma~\ref{lem:ls1}, is at most
    \[4\delta\cost(\opt).\]
\end{proof}

We now use Lemma~\ref{lem:special-cost} to show a clustering scheme of cost 
strictly less than $8\delta\cost(\opt)$.

To this end, we will need the following \emph{Pivot} operation.
Let $\calC_x$, $\calC_y$, $\calC_z$ be three clustering schemes of $G$.
Within each clustering scheme, number the clusters with positive integers.
For each vertex of $G$, we associate it with a 3-dimensional vector
$(x,y,z) \in \mathbb{Z}_+^3$ if it is in $x$-th cluster in $\calC_x$, in the
$y$-th cluster in $\calC_y$, and in the $z$-th cluster in $\calC_z$.
We define the distance between two vertices $u$ and $v$, denoted $d(u,v)$, to be
the Hamming distance of their coordinate vectors.

The operation $\mathrm{Pivot}(\calC_x, \calC_y, \calC_z)$ produces a new clustering scheme
$\calC$ as follows. 
Initially, all vertices of $G$ are unassigned to clusters, that is, we start with $\calC = \emptyset$.
While $\bigcup \calC \neq V(G)$, that is, not all vertices of $G$ are assigned into clusters,
we define a new cluster in $G$ and add to $\calC$.
The new cluster is found as follows: we take a coordinate $(x,y,z)$ with the maximum number of vertices
of $V(G) \setminus \bigcup \calC$ assigned to it (breaking ties arbitrarily),
 and creates a new cluster
consisting of all vertices of $V(G) \setminus \bigcup \calC $ within distance at most $1$
from $(x,y,z)$. 
This concludes the definition of the \emph{Pivot} operation.

The following lemma encapsulates the crux of the analysis of the \emph{Pivot} operation.
\begin{lemma}\label{lem:pivot-crux}
Let $G$ be an unweighted \textsc{Correlation Clustering} instance, let $\calC_x$, $\calC_y$, $\calC_z$ be three clustering schemes, and let $\calC := \mathrm{Pivot}(\calC_x,\calC_y,\calC_z)$.
Divide each of $E^+$ and $E^-$ into 4 sets by the distances. $E^+_i$ denotes the set of 
pairs of $E^+$ with distance $i$, and $E^-_i$ denotes the set of pairs of $E^-$ with distance $i$.
Call pairs in $E^+_0\cup E^+_1\cup E^-_3$ \emph{normal}, and call other pairs \emph{special}.
Then, $(E^+ \cap \Ext(\calC))$ contains no edges of $E^+_0$ and $(E^- \cap \Int(\calC))$ contains
no pairs of $E^-_3$. 
Furthermore, the number of special edges is at least half of the size of $(E^+ \cap \Ext(\calC)) \cup (E^- \cap \Int(\calC))$. %
\end{lemma}
\begin{proof}
    The first claim follows immediately by definition: two vertices with distance $0$ will always be put into the same cluster, while two vertices with distance $3$ will never be put into the same cluster.

    Hence, all pairs in $E^- \cap \Int(\calC)$ are special, and the only
    edges of $E^+ \cap \Ext(\calC)$ that are normal are from $E_1^+$.
Observe that, while creating a cluster with a pivot in $(x,y,z)$,
an edge $uv \in E_1^+$ ends up in $\Ext(\calC)$ if $d(u,(x,y,z)) = 1$
and $d(v,(x,y,z)) = 2$ (or vice versa).
We will charge such edges
to pairs of $E_2 := E_2^- \cup E_2^+$ with at least one endpoint in the newly created cluster.
Note that $E_2$ consists only of special pairs.

Consider a step of the $\mathrm{Pivot}(\calC_x, \calC_y, \calC_z)$ operation
creating a cluster $C$ with pivot $(x,y,z)$.
To finish the proof of the lemma, it suffices to show
that the number of edges $uw \in E^+_1$ where $u \in C$ and $w \notin C$ remains in the graph
is not larger than the number of pairs $u'w' \in E_2$ with $|\{u',w'\} \cap C| \geq 1$ 
(i.e., the number of edges of $E_2$ that will be deleted when deleting the cluster $C$.)

To this end, for coordinate vector $(x',y',z')$, let $n_{(x',y',z')}$ be the number of vertices 
    at coordinate $(x',y',z')$ that are
    unassigned to clusters just before the cluster $C$ is created.

    We say that an edge $uw \in E_1^+$ is \emph{cut} if one endpoint of $uw$ is in $C$,
    and the second endpoint is outside $C$ and unassigned to any cluster at the moment
    of creating $C$.

    If an edge
    from $E_1^+$ is cut, its endpoint not in $C$ has coordinate vector
    at distance $2$ to $(x,y,z)$, say $(x',y',z)$. The number of $E^+_1$ edges cut to this coordinate is at most $v_{(x',y',z)}(v_{(x',y,z)}+v_{(x,y',z)})$. We charge them to the $E_2$ edges between $(x,y,z)$ and $(x',y',z)$ and the $E_2$ edges between $(x',y,z)$ and $(x,y',z)$. Because these
    edge sets are complete, it suffices to show that
    \[ n_{(x',y',z)}(n_{(x',y,z)}+n_{(x,y',z)}) \le n_{(x',y',z)}n_{(x,y,z)}+n_{(x',y,z)}n_{(x,y',z)}.\]
    This easily follows from the fact that $n_{(x,y,z)}\ge\max\{n_{(x',y',z)},n_{(x',y,z)},n_{(x,y',z)}\}$.
    This finishes the proof of the lemma.
\end{proof}
Lemma~\ref{lem:pivot-crux} suggests to bound the number of the special pairs. This is how we will do it.
\begin{lemma}\label{lem:pivot-special}
Let $G$, $\calC_x$, $\calC_y$, $\calC_z$, $\calC$, $E^+_i$, and $E^-_i$ for $i=0,1,2,3$ be as in Lemma~\ref{lem:pivot-crux}.
Then,
  \[ w(E^-_0 \cup E^-_1 \cup E^-_2) \leq w(E^- \cap \Int(\calC_x)) + w(E^- \cap \Int(\calC_y)) + w(E^- \cap \Int(\calC_z)), \]
and
  \[ w(E^+_2 \cup E^+_3) \leq 
    w(E^+ \cap \Ext(\calC_x) \cap \Ext(\calC_y)) + 
    w(E^+ \cap \Ext(\calC_y) \cap \Ext(\calC_z)) + 
    w(E^+ \cap \Ext(\calC_z) \cap \Ext(\calC_x)).\]
\end{lemma}
\begin{proof}
Every pair in $E^-_0\cup E^-_1\cup E^-_2$ is present in at least once 
in $\Int(\calC_x)$, $\Int(\calC_y)$, and $\Int(\calC_z)$.
Every edge in $E^+_2\cup E^+_3$ is present at least once in
$\Ext(\calC_x)\cap\Ext(\calC_y)$,
$\Ext(\calC_y)\cap\Ext(\calC_z)$, and
$\Ext(\calC_z)\cap\Ext(\calC_x)$.
\end{proof}  

We show that if $\ls1$ and $\ls2$ are both not $(2-\delta)$-competitive,
then $\calC := \mathrm{Pivot}(\opt, \ls1, \ls2)$ returns a clustering of cost strictly less
than $8\delta \cost(\opt)$. This will be a contradiction with the optimality of $\opt$
for $\delta = \frac{1}{8}$.
That is, we will use $\calC_x = \opt$, $\calC_y = \ls1$, and $\calC_z = \ls2$.

Combining Lemmata~\ref{lem:special-cost} and~\ref{lem:pivot-special} gives immediately the following.
\begin{lemma}
\label{lem:special-edge}
    If both $\ls1$ and $\ls2$ are not $(2-\delta)$-competitive, then the total weight
    of special pairs is less than $4\delta\cost(\opt)$.
\end{lemma}

The final lemma below is the first place where we use that $w \equiv 1$ in Theorem~\ref{thm:two-ls}. 
This is necessary for Lemma~\ref{lem:pivot-crux} to be useful, as it considers the number of
pairs, not their weight.

\begin{lemma}
\label{lem:pivot-cost}
    Assume $w \equiv 1$.
    If both $\ls1$ and $\ls2$ are not $(2-\delta)$-competitive, then there exists a solution with cost less than $8\delta\cost(\opt)$.
\end{lemma}

\begin{proof}
    Recall that $\calC$ is the clustering scheme being the result of $\mathrm{Pivot}(\opt,\ls1,\ls2)$.
    By Lemma~\ref{lem:pivot-crux}, the number of normal pairs 
    accounted in $\cost(\calC)$ is not larger than the number of special pairs.
    By Lemma~\ref{lem:special-edge} and the assumption $w \equiv 1$,
    the cost of $\calC$ on normal pairs and on special pairs are both less
    than $4\delta\cost(\opt)$. Hence, $\cost(\calC) < 8\delta\cost(\opt)$.
\end{proof}

Since $\opt$ is an optimum solution,
Lemma~\ref{lem:pivot-cost} gives a contraction for $\delta = \frac{1}{8}$.
Hence, at least one of $\ls1$ and $\ls2$ is $(2-\frac{1}{8})$-competitive.
This completes the proof of Theorem~\ref{thm:two-ls}.

\paragraph{A remark on a hard instance.}
We finish this section with an example that the Local-Search-Flip-Local-Search algorithm can be as bad
as $\frac{14}{9}$-competitive. 
Consider a graph of $3\times5\times5$ vertices, where the vertices have coordinates from $(1,1,1)$ to $(3,5,5)$. There is an edge between two vertices if and only if their Hamming distance is at most 2. The optimal solution is to cluster vertices which have the same $x$-coordinate, and the cost is $(3-1)(5+5-1)/2=9$ per vertex. It can be checked that $\ls1$ and $\ls2$ may respectively cluster vertices which have the same $y$-coordinate/$z$-coordinate, and the cost is $(5-1)(3+5-1)/2=14$ per vertex. The ratio is $\frac{14}{9}\approx1.56$.

\subsection{Iterative version of local search}\label{sec:iter-ls}

In this section we present a more complicated iterative local search algorithm
that achieves an approximation ratio of $(2-\frac{2}{13}+\varepsilon)$ for any
$\varepsilon > 0$.
We also present it in full generality, where we are not able to find an exact
local optimum, but only an approximate one. 
We start with defining what this ``approximate one'' means.

Similarly as in the warm-up section, we will solve unweighted \textsc{Correlation Clustering}
instances, but use local search on instances with some of the edge weights increased from the flips. 
The increase from the flips will be small: we will use a constant $\beta > 0$ (set to $\beta = 0.5$ in the end)
and increase the weight of some edges by $\beta$ at most twice. 
Hence, the following definition needs to include weight functions $w : \binom{V(G)}{2} \to \mathbb{Z}_+$. We say that the weight function $w$ is \emph{normal} if $w(uv) = 1$ for $uv \notin E(G)$
and $w(uv) \geq 1$ for $uv \in E(G)$.

\begin{definition}\label{def:apx-max}
Let $G$ be an unweighted \textsc{Correlation Clustering} instance,
let $\varepsilon > 0$ be an accuracy parameter,
and let $(\calK, E_{\adm})$ be a preclustered instance for $G$ returned
by Theorem~\ref{thm:preprocessed-wrapper} for $(G,\varepsilon)$.
That is, $(\calK,E_\adm)$ is $\varepsilon$-good
and let $\calC^*_{(\calK, E_{\adm})}$ be the (unknown) $\varepsilon$-good clustering
scheme in $(\calK, E_{\adm})$ whose existence is promised by Theorem~\ref{thm:preprocessed-wrapper}.

Let $0 < \gamma < \varepsilon^{13}/4$ be a constant and $w : \binom{V(G)}{2} \to \mathbb{Z}_+$
be a normal weight function.
Then, a clustering scheme $\calC$ is a
\emph{$\gamma$-good local optimum for $w$} if
for every cluster $C$ of $\calC^*_{(\calK,E_{\adm})}$ it holds that
\[
\sum_{C\in\calC^*_{(\calK,E_{\adm})}} ( \cost_w(\calC) - \cost_w(\calC + C) )  \le 2 \gamma |E_{\adm}| .    
\]
\end{definition}

We are now ready to present our algorithm. 
It is parameterized by an accuracy parameter $\varepsilon > 0$, 
a constant $0 < \gamma < \varepsilon^{13}/4$, a constant $\beta > 0$, and a number of iterations $k$.
(We will set $\beta = 0.5$ and $0 < \gamma \ll \varepsilon$ to be small constants
 in the end.)
The input consists of an unweighted \textsc{Correlation Clustering} instance
$G$ and 
the algorithm starts with computing
an $\varepsilon$-good preclustered instance $(\calK, E_{\adm})$ for $G$
using Theorem~\ref{thm:preprocessed-wrapper}.

The algorithm uses two operations that are worth recalling. 
First, it explicitly uses the \emph{Pivot} operation described in the previous section.
Second, for a weight function $w$ and clustering scheme $\calC$, it performs the flip by computing
a new weight function $w' := w + \beta (E^+ \cap \Ext(\calC))$, which is a shorthand
for adding weight $\beta$ to all edges of $G$ that connect different clusters of $\calC$.

\begin{tcolorbox}[beforeafter skip=10pt]\label{alg:iterative-local-search}
    \textbf{Iterated-flipping Local Search($G$)}
    \begin{itemize}[itemsep=0pt, topsep=4pt]
        \item Let $(\calK,E_{\adm})$ be the preclustering obtained via Theorem~\ref{thm:preprocessed-wrapper} on $G$ and $\varepsilon$.
        \item Let $w_0 \equiv 1$ be the uniform weight function.
        \item Let $\calC_0'$ be a $\gamma$-good local optimum for $w_0$.
        \item For $i=1,2,\ldots,k$
        \begin{itemize}
            \item $w_i := w_0 + \beta (E^+ \cap \Ext(\calC_{i-1}'))$.
            \item Let $\calC_i$ be a $\gamma$-good local optimum for $w_i$.
            \item $w_i' := w_i + \beta (E^+ \cap \Ext(\calC_i))$.
            \item Let $\calC_i'$ be a $\gamma$-good local optimum for $w_i'$.
            \item $\calC_i'' := \mathrm{Pivot}(\calC_{i-1}',\calC_i,\calC_{i}')$ \footnote{Pivot can be implemented in $O(n)$ time and space. }
        \end{itemize}
        Output best of $\calC_1,\ldots,\calC_k,\calC_0',\ldots,\calC_k',\calC_1'',\ldots,\calC_k''$ with respect to the original cost function.
    \end{itemize}
\end{tcolorbox}

We prove in this section the following theorem.

\begin{theorem}
\label{thm:it-ls}
  For every $0 \leq \alpha < \frac{2}{13}$
  there exists a positive integer $k$ and a real $\varepsilon_0 > 0$
  such that for every $0 < \varepsilon \leq \varepsilon_0$ and for every $0 < \gamma < \varepsilon^{13}/4$,
  the Iterated-flipping Local Search Algorithm
  returns a $(2-\alpha)$-approximation to \textsc{Correlation Clustering} on an unweighted input $G$
  if run on parameters $\varepsilon$, $\gamma$, $\beta = 0.5$, and $k$.
\end{theorem}

Let $\varepsilon_0' > 0$ be sufficiently small such that 
every $0 \leq \varepsilon \leq \varepsilon_0'$ is fine for Theorem~\ref{thm:preprocessed-wrapper}
to work. 
We set $\varepsilon_0$ and $k$ as follows:
\[ \varepsilon_0 = \min\left(\varepsilon_0', \frac{1}{12} \left(\frac{2}{13}-\alpha\right)\right) \quad\mathrm{and}\quad k = 1 + \left\lceil \frac{2}{\frac{2}{13}-\alpha} \right\rceil. \]
Let $0 < \varepsilon \leq \varepsilon_0$ and $0 < \gamma < \varepsilon^{13}/4$ be arbitrary. 

Let $(\calK,E_{\adm})$ the $\varepsilon$-good preclustered instance returned
by Theorem~\ref{thm:preprocessed-wrapper}
and let $\calC^\ast_{(\calK,E_{\adm})}$ be the $\varepsilon$-good clustering
scheme promised by Theorem~\ref{thm:preprocessed-wrapper}.

Recall that by $\cost(\calC)$ we mean the cost of $\calC$ with regards to the uniform
weight $w_0 \equiv 1$. If we want to use a different weight function $w$, we denote it
by $\cost_w(\calC)$.

We set $\hat{\alpha} := (\frac{2}{13}+\alpha)/2$.
We will aim at a solution of cost $(2-\hat{\alpha})\cost(\calC^\ast_{(\calK,E_{\adm})}))$.
By the properties promised by Theorem~\ref{thm:preprocessed-wrapper}
and by our choice of $\varepsilon_0$,
this is enough to obtain a $(2-\alpha)$-approximation.
Hence, for contradiction we assume that the costs of
all clustering schemes $\calC_i$, $\calC_i'$, and $\calC_i''$ generated by the algorithm are larger than $(2-\hat{\alpha})\cost(\calC^\ast_{(\calK,E_{\adm})})$. 
We will reach a contradiction.

We start with an analog of Lemma~\ref{lem:prop-ls}.

\begin{lemma}\label{lem:weighted-local-search}
  Let $\calC$ be a $\gamma$-good local optimum for $G$ and a weight
  function $w$. Then, 
    \begin{align*}
        &w(E^- \cap \Int(\calC^\ast_{(\calK,E_{\adm})}))+2w(E^+ \cap \Ext(\calC^\ast_{(\calK,E_{\adm})})) + \varepsilon \cost(\calC^\ast_{(\calK,E_{\adm})}) \ge\\ 
        & w(E^+\cap\Ext(\calC)\cap\Int(\calC^\ast_{(\calK,E_{\adm})}))+2w(E^+\cap\Ext(\calC)\cap\Ext(\calC^\ast_{(\calK,E_{\adm})}))\\
        &+w(E^-\cap\Int(\calC)\cap\Int(\calC^\ast_{(\calK,E_{\adm})}))+2w(E^-\cap\Int(\calC)\cap\Ext(\calC^\ast_{(\calK,E_{\adm})}))
    \end{align*}
\end{lemma}
\begin{proof}
Let $C$ be a cluster of $\calC^\ast_{(\calK,E_{\adm})}$. By the optimality of $\calC$, we have
\[ \sum_{C \in \calC^\ast_{(\calK,E_{\adm})}}  \cost_w(\calC + C) + 2\gamma |E_{\adm}| \geq \sum_{C \in \calC^\ast_{(\calK,E_{\adm})}} \cost_w(\calC). \]
Equivalently,
 \begin{align*}
 & \sum_{C \in \calC^\ast_{(\calK,E_{\adm})}} ( w(E^+ \cap \Ext(C)) + w(E^- \cap \Int(C)) ) + 2\gamma |E_{\adm}| \geq\\
 & \sum_{C \in \calC^\ast_{(\calK,E_{\adm})}} ( w(E^+ \cap \Ext(\calC) \cap \Int(C)) + w(E^+ \cap \Ext(\calC) \cap \Ext(C))\\
  & \qquad\qquad\quad + w(E^- \cap \Int(\calC) \cap \Int(C)) + w(E^- \cap \Int(\calC) \cap \Ext(C)) ).
\end{align*}
The lemma follows from observing that
\begin{align*}
  2\gamma |E_{\adm}|  \leq \frac{4\gamma}{\varepsilon^{12}} \cost(\calC^\ast_{(\calK,E_{\adm})}) 
    \leq \varepsilon \cost(\calC^\ast_{(\calK,E_{\adm})}).
\end{align*}
Here, the first inequality uses the properties of $(\calK,E^{\adm})$ promised
by Theorem~\ref{thm:preprocessed-wrapper} while the last inequality
uses the assumption $\gamma < \varepsilon^{13}/4$.
\end{proof}

We now proceed to the analogs of Lemmata~\ref{lem:ls1} and~\ref{lem:ls2}.
\begin{lemma}\label{lem:iterated-M-function}
Let $\calC_1,\ldots,\calC_\ell$ be clustering schemes
and 
\[w := w_0 + \sum_{i=1}^\ell \beta (E^+ \cap \Ext(\calC_i)). \]
(That is, for every $1 \leq i \leq \ell$, we add a weight of $\beta$ to every edge
 connecting two distinct clusters of $\calC_i$.)
Let $\calC$ be a $\gamma$-good local optimum of $G$ and $w$.
If $\cost(\calC) > (2-\hat{\alpha}) \cost(\calC^\ast_{(\calK,E_{\adm})})$, then 

    \begin{align*}
        \left(\hat{\alpha} + \varepsilon\right) \cost(\calC^\ast_{(\calK,E_{\adm})}) &+ 2\beta \left(\sum_{i=1}^k|E^+ \cap \Ext(\calC^\ast_{(\calK,E_{\adm})})\cap \Ext(\calC_i)|\right) \ge\\
        &\beta\sum_{i=1}^k|E^+ \cap \Ext(\calC)\cap \Ext(\calC_i)|\\
        &+|E^+ \cap \Ext(\calC)\cap\Ext(\calC^\ast_{(\calK,E_{\adm})})| \\
        &+|E^- \cap \Int(\calC)|.
    \end{align*}
\end{lemma}
\begin{proof}
    The assumption $\cost(\calC) > (2-\hat{\alpha})\cost(\calC^\ast_{(\calK,E_{\adm})})$ expands as
    \[ |E^+ \cap \Ext(\calC)| + |E^- \cap \Int(\calC)| \ge (2-\hat{\alpha})\left(|E^+ \cap \Ext(\calC^\ast_{(\calK,E_{\adm})})| + |E^- \cap \Int(\calC^\ast_{(\calK,E_{\adm})})|\right). \]
    Rearranging this, we get 
    \begin{align*} 
    \hat{\alpha} \cost(\calC^\ast_{(\calK,E_{\adm})}) &\ge 2|E^+ \cap \Ext(\calC^\ast_{(\calK,E_{\adm})})| + 2|E^- \cap \Int(\calC^\ast_{(\calK,E_{\adm})})| \\&\quad- |E^+ \cap \Ext(\calC)| - |E^- \cap \Int(\calC)|. \end{align*}
    Applying now the definition of $w(\calC^\ast_{(\calK,E_{\adm})})$, we get
    \begin{align*}
        \hat{\alpha} \cost(\calC^\ast_{(\calK,E_{\adm})}) &\ge 2w(\calC^\ast_{(\calK,E_{\adm})}) - 2\beta\sum_{i=1}^k|E^+ \cap \Ext(\calC^\ast_{(\calK,E_{\adm})})\cap\Ext(C_i)| \\
        & - |E^+ \cap \Ext(\calC)| - |E^- \cap \Int(\calC)|.
    \end{align*}
    Adding the inequality of Lemma~\ref{lem:weighted-local-search} to both sides, we obtain
    \begin{align*}
        &\left(\hat{\alpha} + \varepsilon\right)\cost(\calC^\ast_{(\calK,E_{\adm})}) + 2\beta\sum_{i=1}^k|E^+ \cap \Ext(\calC^\ast_{(\calK,E_{\adm})})\cap\Ext(C_i)| \ge \\
        &\quad w(E^+\cap\Ext(\calC)\cap\Int(\calC^\ast_{(\calK,E_{\adm})}))+2w(E^+\cap\Ext(\calC)\cap\Ext(\calC^\ast_{(\calK,E_{\adm})}))\\
        &\quad +w(E^-\cap\Int(\calC)\cap\Int(\calC^\ast_{(\calK,E_{\adm})}))+2w(E^-\cap\Int(\calC)\cap\Ext(\calC^\ast_{(\calK,E_{\adm})}))  \\
        &\quad +w(E^- \cap \Int(\calC^\ast_{(\calK,E_{\adm})}))- |E^+ \cap \Ext(\calC)| - |E^- \cap \Int(\calC)|.
    \end{align*}
    The lemma now follows from the above and the following three simple observations:
    \begin{align*}
    &w(E^-\cap\Int(\calC)\cap\Int(\calC^\ast_{(\calK,E_{\adm})}))+2w(E^-\cap\Int(\calC)\cap\Ext(\calC^\ast_{(\calK,E_{\adm})})) \\
    &\quad = |E^- \cap \Int(\calC)| + |E^- \cap \Int(\calC) \cap \Ext(\calC^\ast_{(\calK,E_{\adm})})|,\\
    &   w(E^- \cap \Int(\calC^\ast_{(\calK,E_{\adm})})) = |E^- \cap \Int(\calC^\ast_{(\calK,E_{\adm})})| \geq |E^- \cap \Int(\calC) \cap \Int(\calC^\ast_{(\calK,E_{\adm})})|,
    \end{align*}
    and
    \begin{align*}
    & w(E^+\cap\Ext(\calC)\cap\Int(\calC^\ast_{(\calK,E_{\adm})}))+2w(E^+\cap\Ext(\calC)\cap\Ext(\calC^\ast_{(\calK,E_{\adm})}))\\
     &\quad = w(E^+ \cap \Ext(\calC)) + w(E^+ \cap \Ext(\calC) \cap \Ext(\calC^\ast_{(\calK,E_{\adm})}))\\
     &\quad \geq w(E^+ \cap \Ext(\calC)) + |E^+ \cap \Ext(\calC) \cap \Ext(\calC^\ast_{(\calK,E_{\adm})})|\\
     &\quad = |E^+ \cap \Ext(\calC)| 
     + \beta \sum_{i=1}^\ell |E^+ \cap \Ext(\calC) \cap \Ext(\calC_i)|
     + |E^+ \cap \Ext(\calC) \cap \Ext(\calC^\ast_{(\calK,E_{\adm})})|.
       \end{align*}
\end{proof}

Fix $1 \leq i \leq k$.
By Lemmata~\ref{lem:pivot-crux} and~\ref{lem:pivot-special}, we have
\begin{align}
\cost(\calC_i'') &\leq 2\left(|E^- \cap \Int(\calC_{i-1}')| + |E^- \cap \Int(\calC_i)| + |E^- \cap \Int(\calC_i')|\right) \nonumber\\
                 &\quad+ 2\big(
                    |E^+ \cap \Ext(\calC_{i-1}') \cap \Ext(\calC_i)| + 
                    |E^+ \cap \Ext(\calC_{i-1}') \cap \Ext(\calC_i')|\nonumber\\
                 &\quad\qquad +
                    |E^+ \cap \Ext(\calC_i) \cap \Ext(\calC_i')|\big).\label{eq:bnd-pivot}
\end{align} 

Recall $\beta = 0.5$. 
Lemma~\ref{lem:iterated-M-function} applied to $\calC_i$ gives
\begin{align}
& \left(\hat{\alpha}+\varepsilon\right)\cost(\calC^\ast_{(\calK,E_{\adm})}) + |E^+ \cap \Ext(\calC^\ast_{(\calK,E_{\adm})}) \cap \Ext(\calC_{i-1}')| \nonumber\\
& \quad \geq 0.5 |E^+ \cap \Ext(\calC_i) \cap \Ext(\calC_{i-1}')| \nonumber\\
&\quad + |E^+ \cap \Ext(\calC^\ast_{(\calK,E_{\adm})}) \cap \Ext(\calC_i)| + 
           |E^- \cap \Int(\calC_i)|.\label{eq:bnd-Ci}
\end{align}
Lemma~\ref{lem:iterated-M-function} applied to $\calC_i'$ gives
\begin{align}
& \left(\hat{\alpha}+\varepsilon\right)\cost(\calC^\ast_{(\calK,E_{\adm})}) + |E^+ \cap \Ext(\calC^\ast_{(\calK,E_{\adm})}) \cap \Ext(\calC_{i-1}')| + |E^+ \cap \Ext(\calC^\ast_{(\calK,E_{\adm})}) \cap \Ext(\calC_i)| \nonumber\\
& \quad \geq 0.5 |E^+ \cap \Ext(\calC_i') \cap \Ext(\calC_{i-1}')| + 0.5 |E^+ \cap \Ext(\calC_i') \cap \Ext(\calC_i)| \nonumber\\
& \quad + |E^+ \cap \Ext(\calC^\ast_{(\calK,E_{\adm})}) \cap \Ext(\calC_i')| + |E^- \cap \Int(\calC_i')|.\label{eq:bnd-Cii}
\end{align}
To bound the right hand side of~\eqref{eq:bnd-pivot}, we add twice~\eqref{eq:bnd-Ci} with twice~\eqref{eq:bnd-Cii}, obtaining the following.
\begin{align*}
& \left(4\hat{\alpha}+4\varepsilon\right)\cost(\calC^\ast_{(\calK,E_{\adm})}) + 4|E^+ \cap \Ext(\calC^\ast_{(\calK,E_{\adm})}) \cap \Ext(\calC_{i-1}')| \\
& \quad \geq 2|E^- \cap \Int(\calC_i)| + 2|E^- \cap \Int(\calC_i')| \\
& \quad + |E^+ \cap \Ext(\calC_{i-1}') \cap \Ext(\calC_i)| + 
                    |E^+ \cap \Ext(\calC_{i-1}') \cap \Ext(\calC_i')|\\
                 &\quad +
                    |E^+ \cap \Ext(\calC_i) \cap \Ext(\calC_i')| + 2|E^+ \cap \Ext(\calC^\ast_{(\calK,E_{\adm})}) \cap \Ext(\calC_i')|.
\end{align*}
Combining the above with~\eqref{eq:bnd-pivot}, we obtain
\begin{align*}
& \left(4\hat{\alpha}+4\varepsilon\right)\cost(\calC^\ast_{(\calK,E_{\adm})}) + |E^- \cap \Int(\calC_{i-1}')| + 4|E^+ \cap \Ext(\calC^\ast_{(\calK,E_{\adm})}) \cap \Ext(\calC_{i-1}')|\\
& \quad \geq 0.5 \cost(\calC_i'') + |E^- \cap \Int(\calC_i')| + 2|E^+ \cap \Ext(\calC^\ast_{(\calK,E_{\adm})}) \cap  \Ext(\calC_{i}')|.
\end{align*}
Using the assumption $\cost(\calC_i'') > (2-\hat{\alpha})\cost(\calC^\ast_{(\calK,E_{\adm})})$, we get
\begin{align}
& \left(4.5\hat{\alpha}+4\varepsilon - 1\right)\cost(\calC^\ast_{(\calK,E_{\adm})}) + |E^- \cap \Int(\calC_{i-1}')| + 4|E^+ \cap \Ext(\calC^\ast_{(\calK,E_{\adm})}) \cap \Ext(\calC_{i-1}')|\nonumber\\
& \quad > |E^- \cap \Int(\calC_i')| + 2|E^+ \cap \Ext(\calC^\ast_{(\calK,E_{\adm})}) \cap \Ext(\calC_{i}')|.\label{eq:bnd-progress}
\end{align}
The above estimate motivates us to define the following value for $0 \leq i \leq k$:
\[ b_i := \frac{|E^- \cap \Int(\calC_i')| + 2|E^+ \cap \Ext(\calC^\ast_{(\calK,E_{\adm})}) \cap \Ext(\calC_i')|}{\cost(\calC^\ast_{(\calK,E_{\adm})})}. \]
Rewriting~\eqref{eq:bnd-progress}, we obtain
\begin{equation}\label{eq:bnd-progress2}
b_{i-1} + 4.5\hat{\alpha}+4\varepsilon - 1 > b_i - b_{i-1}.
\end{equation}
Applying Lemma~\ref{lem:iterated-M-function} to $\calC_0'$, we have
\[ \left(\hat{\alpha} + \varepsilon\right) \cost(\calC^\ast_{(\calK,E_{\adm})}) \geq |E^+ \cap \Ext(\calC^\ast_{(\calK,E_{\adm})}) \cap \Ext(\calC_0')| + |E^- \cap \Int(\calC_0')|. \]
Hence,
\begin{equation}\label{eq:bnd-b0}
b_0 \leq 2\hat{\alpha} + 2\varepsilon < 1.
\end{equation}
Since every $b_i$ is nonnegative, by the choice of $k$ there exists $i$ such that $b_i > b_{i-1} - (\frac{2}{13}-\hat{\alpha})$. 
Let $1 \leq i_0 \leq k$ be minimum such that
\[ b_{i_0} > b_{i_0-1} - \left(\frac{2}{13}-\hat{\alpha}\right). \]
By~\eqref{eq:bnd-b0} and the definition of $i_0$, for every $0 \leq i < i_0$ it holds that
\[ b_i \leq 2\hat{\alpha} + 2\varepsilon. \]
Hence, by combining~\eqref{eq:bnd-progress2} with the definition of $i_0$, we have
\[ 6.5 \hat{\alpha} + 6\varepsilon - 1 > \hat{\alpha} - \frac{2}{13}. \]
Equivalently,
\[ 5.5 \hat{\alpha} + 6\varepsilon > \frac{11}{13}. \]
Plugging in the definition of $\hat{\alpha}$,
\[ 5.5 \alpha + 12\varepsilon > \frac{11}{13}. \]
This is a contradiction with the assumption $\alpha < \frac{2}{13}$ and the choice of $\varepsilon_0$. 
This finishes the proof of Theorem~\ref{thm:it-ls}.

\newcommand{\nearopt}{\calC^*_{(\calK,E_{\adm})}}

\section{Polynomial-Time Implementation of Local Search}

In this section we provide a polynomial-time implementation of the local search routine
needed in Section~\ref{sec:iter-ls}.
That is, given a \textsc{Correlation Clustering} instance $G$
with a preclustering $(\calK, E_{\adm})$ obtained via Theorem~\ref{thm:preprocessed-wrapper}
(with the implicit $\varepsilon$-good clustering scheme $\nearopt$)
and a weight function $w$, we are to find a $\gamma$-good local optimum for $w$. 

The usage in Section~\ref{sec:iter-ls} takes $\varepsilon > 0$ 
as a sufficiently small accuracy parameter 
and $0 < \gamma < \varepsilon^{13}/4$ as a second parameter.
The weight function $w$ is equal to $1$ on non-edges and takes values $w(e) \in [1,2]$
for edges $e \in E^+$. 
It will be a bit cleaner to consider a slightly more general variant
where $w(e) \in [1,W]$ for a constant parameter $W \geq 1$, but one can keep in mind
that we actually only need the case $W=2$.
In our algorithm, the weight of a pair can be computed in constant time by checking whether they are neighbors and whether they are in the same cluster in some previous local search solution(s).

\begin{lemma}
\label{lem:cost-adm}
Let $w$ be a weight function that equals $1$ on nonedges and takes values in $[1,W]$ on edges.
Consider two clusterings $\calC$ and $\calC'$ and assume that for each cluster $C_i\in\calC'$, 
\begin{equation}\label{eq:c-optimal-cprime}
\cost_w(\calC + C_i)  + \gamma \sum_{u \in C_i} d^{adm}(u) + \Delta_i \ge \cost_w(\calC).    
\end{equation}
 Then $\cost_w(\calC) \le 2\cost_w(\calC') + \varepsilon \cost(\nearopt) + \sum_i\Delta_i$.
\end{lemma}
\begin{proof}
Equivalently to \eqref{eq:c-optimal-cprime} we have
 \begin{align*}
 & w(E^+ \cap \Ext(C_i)) + |E^- \cap \Int(C_i)| + \gamma \sum_{u \in C_i} d^{\adm}(u) + \Delta_i \geq\\
 & \quad 
   w(E^+ \cap \Ext(\calC) \cap \Int(C_i)) + w(E^+ \cap \Ext(\calC) \cap \Ext(C_i))\\
  & \quad+ |E^- \cap \Int(\calC) \cap \Int(C_i)| + |E^- \cap \Int(\calC) \cap \Ext(C_i)|.
\end{align*}
The theorem follows from summing the above inequality for every cluster $C_i$ of $\calC'$
and observing that
\begin{align*}
 \sum_{C_i \in \calC'} \gamma \sum_{u \in C} d^{\adm}(u) &= \gamma \sum_{u \in V(G)} d^{\adm}(u)
    = 2\gamma |E^{\adm}| \\
   & \leq \frac{4\gamma}{\varepsilon^{12}} \cost(\nearopt) 
    \leq \varepsilon \cost(\nearopt).
\end{align*}
Here, the penultimate inequality uses the properties of $(\calK,E^{\adm})$ promised
by Theorem~\ref{thm:preprocessed-wrapper} while the last inequality
uses the assumption $\gamma < \varepsilon^{13}/4$.
\end{proof}

We will need the following weight-adjusted variants of previously introduced graph notation.
\begin{definition}
    Let $d_w(v) = \sum_{u \mid (u, v) \in E^+}w(u, v)$ be the weighted degree of $v$ under $w$. 
    Similarly, we define
    \[
        d_w(v, S) = \sum_{u \mid u \in S \land (u, v) \in E^+}w(u, v).
    \]
    where $S$ is a multiset of vertices. 
\end{definition}

\begin{observation}
    The following statements follow directly from the properties of $(\calK,E_{\adm})$.
    \begin{itemize}
        \item $\forall v \in V, d(v) \leq d_w(v) \leq W d(v)$. 
        \item $d^{\adm}(v) \leq 2\varepsilon^{-3}d_w(v)$. 
        \item If $(u, v)$ is an admissible pair, then $\frac{1}{2W}\varepsilon d_w(v) \leq d_w(u) \leq 2W\varepsilon^{-1}d_w(v)$. 
    \end{itemize}    
\end{observation}

We define the neighborhood function $N(v)$ as follows, which is slightly different from $N_v$, the admissible neighborhood. 
\[
N(v) = \begin{cases}
    N_v \setminus \bigcup_{K \in \calK} K & \mathrm{if\ }v\mathrm{\ does\ not\ belong\ to\ any\ atom,}\\
    K \cup \left(\bigcap_{u \in K} N_u \right) & \mathrm{if\ }v\mathrm{\ belongs\ to\ an\ atom\ }K.
\end{cases}
\]
We define
\[
    K(v) = \begin{cases}
            \{v\} & \mathrm{if\ }v\mathrm{\ does\ not\ belong\ to\ any\ atom}, \\
            K' & \mathrm{if\ }v\mathrm{\ belongs\ to\ an\ atom\ }K'.
    \end{cases}
\]
\[
    D(v) = N(v) \setminus K(v) ~. 
\]

\paragraph{Estimates of set sizes.}
We will need a few estimates of sizes of various neighborhood-like sets. 

\begin{lemma}
    \label{lem:size-of-neighborhood}
    For all $v \in V$, $|N(v)| \leq 6\varepsilon^{-3}d(v)$. 
\end{lemma}

\begin{proof}
    Since $|N_v| = d^{adm}(v) + 1 \leq 2\varepsilon^{-3}d(v) + 1$, we only need to upper bound $|K|$ when $v$ belongs to atom $K$. 
    According to the \cref{def:precluster}, for sufficient small $\epsilon$, at least half of vertices in $K$ are neighbors of $v$, so $|K| \leq 2d(v) + 1$. 
    Therefore $|N(v)| \leq |N_v| + |K| \leq 6\varepsilon^{-3}d(v)$. 
\end{proof}

\begin{lemma}
    \label{lem:neighbor-degree-similarity}
    For all $r \in V, v \in N(r)$, $|N(r)| \leq 12\varepsilon^{-4}d(v)$. 
\end{lemma}

\begin{proof}
    Since $v \in N(r)$, either $v$ and $r$ are in the same cluster, or $v$ is admissible to $r$. 
    In the former case, $N(r) = N(v)$, according to lemma~\ref{lem:size-of-neighborhood}, $|N(r)| = |N(v)| \leq 6\varepsilon^{-3}d(v)$. 
    In the latter case, we have degree similarity between $r$ and $v$. According to lemma~\ref{lem:size-of-neighborhood}, $|N(r)| \leq 6\varepsilon^{-3}d(r) \leq 12\varepsilon^{-4}d(v)$. 
\end{proof}

All clusters in the clustering scheme maintained by our algorithm will be \emph{almost good clusters} in the following sense. 

\begin{definition}[Almost good clusters]
    We say a cluster $C$ is almost good if there exists a vertex $r\in V$ (not necessarily in $C$) such that $C\subseteq N(r)$.
\end{definition}

\begin{lemma}
    \label{lem:almost-good-size-bound}
    For any vertex $v$ and any almost good cluster $C$ containing $v$, $|C| \leq 12\varepsilon^{-4}d(v)$. 
\end{lemma}

\begin{proof}
    Since $C$ is almost good, there exists $r \in V$ such that $C \subseteq N(r)$. 
    According to lemma~\ref{lem:neighbor-degree-similarity}, $|C| \leq |N(r)| \leq 12\varepsilon^{-4}d(v)$. 
\end{proof}

Since every cluster considered by our local search algorithm will be an almost good cluster, 
      we have the following corollary.

\begin{corollary}
    \label{cor:bounded-minus-cost}
    Let $\calC$ be a clustering maintained by the local search algorithm, then for any $C \in \calC, v \in C$, $|C| \leq 12\varepsilon^{-4}d(v)$. 
\end{corollary}

The last estimate we need is the following. 

\begin{lemma} 
\label{lem:D-times-N-4.2}
    For any vertex $v$ and any $(\varepsilon, \varepsilon/2)$-good cluster $C$ containing $v$, 
    \[
        \frac{1}{576}\varepsilon^{8}|D(v)|\cdot |N(v)| \le \sum_{u \in C} d^{adm}(u) ~.
    \]
\end{lemma}

\begin{proof}
    Since $C$ is a $(\varepsilon, \varepsilon/2)$-good cluster, we have $|C| \geq \frac{1}{2}\varepsilon d(v) \geq \frac{1}{12}\varepsilon^4 |N(v)|$. 
    The second inequality follows from lemma~\ref{lem:size-of-neighborhood}. 
    
    If $|K(v)| \geq \alpha |N(v)|$, then 
    \[
        \sum_{u \in C}d^{\adm}(u) \geq \sum_{u \in K(v)}d^{\adm}(u) \geq \alpha|N(v)|\cdot |D(v)| ~.
    \]
    If $|K(v)| \leq \alpha |N(v)|$, then 
    \begin{align*}
        \sum_{u \in C}d^{\adm}(u) & \geq \sum_{u \in C \setminus K(v)}d^{\adm}(u) \\
        & \geq (|C| - |K(v)|)(|C| - 1) \\
        & \geq \frac{1}{2}|C|(|C| - |K(v)|) \\ 
        & \geq \frac{1}{24}\varepsilon^4\left(\frac{1}{12}\varepsilon^4 - \alpha\right)|N(v)|^2
    \end{align*}
    Let $\alpha = \frac{1}{24}\varepsilon^8$, the lemma always holds. 
\end{proof}

\paragraph{Costs and estimated costs.}
As discussed in the introduction, we try to find a cluster that improves our local search solution
with the help of sampling. We will sample a number of vertices from the sought cluster
and then, for every other vertex $v$, we try to guess whether $v$ is in the sought cluster 
by looking at how many vertices from the sample are adjacent to $v$. 
To this end, we will need the following definitions.

\begin{definition}
Consider a clustering $\calC$ of the graph and let $v$ be a vertex, $C(v) \in \calC$ be its cluster and $K$ be an arbitrary set of vertices. 
Under weight $w$, we define $\rmCostStays_w(K,v)$ to be the total weight of edges $(u,v)$ that are violated in $\calC + (K\setminus \{v\})$ and $\rmCostMoves_w(K,v)$ to be the total weight of edges $(u,v)$ that are violated in $\calC+(K\cup\{v\})$. 
By definition, 
\begin{enumerate}
    \item $\rmCostStays_w(K,v) = d_w(v) - d_w(v, C(v)) + d_w(v, C(v) \cap K) + |C(v)| - |C(v) \cap K| - d(v, C(v)) + d(v, C(v) \cap K) - 1$. %
    \item $\rmCostMoves_w(K,v) = d_w(v) - d_w(v, K) + |K| - d(v,K) - 1$. 
\end{enumerate}
\end{definition}

\begin{lemma}
\label{lem:moves_estimate}
Let $\eta_0 > 1$. 
Consider a clustering $\calC$ of the graph and let $v$ be a vertex, $C(v)$ be its cluster and $K$ be an arbitrary set of vertices of size $s$. 
Furthermore, let $S = \{u_1, \dots, u_{\eta_0}\}$ be a sequence of uniform sample of $K$. 
Consider the clustering $\calC+K$
We define the following random variables
\begin{itemize}
    \item $\rmEstCostStays_w(S,s,v)$ := $d_w(v) - d_w(v, C(v)) - 1 + \frac{s}{\eta_0}\sum_{i = 1}^{\eta_0}[u_i \in C(v)](d_w(v, \{u_i\}) + d(v, \{u_i\}) - 1)$. 
    \item $\rmEstCostMoves_w(S,s,v)$ := $d_w(v) + s - 1 - \frac{s}{\eta_0}\sum_{i = 1}^{\eta_0}(d_w(v, \{u_i\}) + d(v, \{u_i\}))$. 
\end{itemize}
where $[\text{boolean expression}]$ is the indicator function. 

We have that
\begin{enumerate}
    \item ExpCostStays($s,v$) := $\mathbb{E}_S[\text{EstCostStays($S,s,v$)}] = $ CostStays, and  
    \item ExpCostMoves($s,v$) := $\mathbb{E}_S[\text{EstCostMoves($S,s,v$)}] = $ CostMoves. 
\end{enumerate}

\end{lemma}

The next two lemmata say that looking at the estimated costs indeed approximates the real cost well,
    even if we know the actual size of the cluster only approximately.

\begin{lemma}
\label{lem:concentration-weighted}
Let $\eta_0 = \eta^5 > 0$. Consider the setting of Lemma~\ref{lem:moves_estimate}, a vertex $v$ in cluster $C(v)$, an arbitrary cluster
$K$ of size $s$ and a sequence $S$ of random uniform sample of $K$ of length $\eta_0$.
Then, with probability at least  $1 - 4\exp\left(-\frac{1}{2}\eta\right)$ we have that the following
two inequalities hold:
\begin{equation}
    \label{eq:coststay-concentration}
    \rmEstCostStays_w(S,s,v) \in \rmExpCostStays_w(s,v) \pm \frac{1}{\eta^2}Ws
\end{equation}
\begin{equation}
    \label{eq:costmove-concentration}
    \rmEstCostMoves_w(S,s,v) \in \rmExpCostMoves_w(s,v) \pm \frac{1}{\eta^2}Ws
\end{equation}
\end{lemma}
\begin{proof}
    Let $X_i = \frac{s}{\eta_0}[u_i \in C(v)](d_w(v, \{u_i\}) + d(v, \{u_i\}) - 1)$ be an random variable, we have $X_i \in \left[-\frac{s}{\eta_0}, \frac{s}{\eta_0}W\right]$. 
    \Cref{eq:coststay-concentration} holds if and only if 
    \[
        \left|\sum_{i = 1}^{\eta_0}(X_i - E[X_i])\right| \leq \frac{1}{\eta^2}Ws ~.
    \]
    According to Hoeffding's inequality, we have
    \begin{align*}
        \Pr\left[\left|\sum_{i = 1}^{\eta_0}(X_i - E[X_i])\right| \geq \frac{1}{\eta^2}Ws\right] & \leq 2\exp\left(-\frac{2\left(\frac{1}{\eta^2}Ws\right)^2}{\eta_0\left(\frac{s}{\eta_0}(W + 1)^2\right)}\right) \\
        & \leq 2\exp\left(-\frac{2\eta_0W^2}{\eta^4(W+1)^2}\right) \\
        & \leq 2\exp\left(-\frac{1}{2}\eta\right)
    \end{align*}
    So \cref{eq:coststay-concentration} holds with probability at least $1 - 2\exp\left(-\frac{1}{2}\eta\right)$. 

    Similarly, let $Y_i = \frac{s}{\eta_0}(d_w(v, \{u_i\}) + d(v, \{u_i\}))$. \Cref{eq:costmove-concentration} holds if and only if 
    \[
        \left|\sum_{i = 1}^{\eta_0}(Y_i - E[Y_i])\right| \leq \frac{1}{\eta^2}Ws ~.
    \]
    According to Hoeffding's inequality, we have
    \begin{align*}
        \Pr\left[\left|\sum_{i = 1}^{\eta_0}(Y_i - E[Y_i])\right| \geq \frac{1}{\eta^2}Ws\right] & \leq 2\exp\left(-\frac{2\left(\frac{1}{\eta^2}Ws\right)^2}{\eta_0\left(\frac{s}{\eta_0}(W + 1)^2\right)}\right) \\ & \leq 2\exp\left(-\frac{2\eta_0W^2}{\eta^4(W+1)^2}\right) \\
        & \leq 2\exp\left(-\frac{1}{2}\eta\right)
    \end{align*}
    So \cref{eq:costmove-concentration} holds with probability at least $1 - 2\exp\left(-\frac{1}{2}\eta\right)$. 

    By union bound, both \cref{eq:costmove-concentration} and \cref{eq:coststay-concentration} hold with probability at least $1 - 4\exp\left(-\frac{1}{2}\eta\right)$. 
\end{proof}

\begin{lemma}
\label{lem:concentration-weighted-noisy}
Let $\eta_0 = \eta^5 > 0$. Consider the setting of Lemma~\ref{lem:moves_estimate}, a vertex $v$ in cluster $C(v)$, an arbitrary cluster
$K$ of size $s$ and a sequence $S$ of random uniform sample of $K$ of length $\eta_0$. 
Let $\tilde{s} \in (1 \pm \epsilon')s$. 
Then, with probability at least  $1 - 4\exp\left(-\frac{1}{2}\eta\right)$ we have that the following
two inequalities hold:
\begin{equation}
    \label{eq:coststay-concentration-error}
    \rmEstCostStays_w(S,\tilde{s},v) \in \rmExpCostStays_w(s,v) \pm \left(\frac{1}{\eta^2}Ws + 2\epsilon'Ws\right)
\end{equation}
\begin{equation}
    \label{eq:costmove-concentration-error}
    \rmEstCostMoves_w(S,\tilde{s},v) \in \rmExpCostMoves_w(s,v) \pm \left(\frac{1}{\eta^2}Ws + 2\epsilon'Ws\right)
\end{equation}
\end{lemma}

\begin{proof}
    \begin{align*}
        |\rmEstCostStays_w(S,s,v) - \rmEstCostStays_w(S,\tilde{s},v)| & = \left|\frac{s - \tilde{s}}{\eta_0}\sum_{i = 1}^{\eta_0}[u_i \in C(v)](d_w(v, \{u_i\}) + d(v, \{u_i\}) - 1)\right| \\
        & \leq \frac{|s - \tilde{s}|}{\eta_0}\sum_{i = 1}^{\eta_0}|d_w(v, \{u_i\}) + d(v, \{u_i\}) - 1| \\
        & \leq \frac{\epsilon's}{\eta_0}\eta_0W \\
        & < 2\epsilon'Ws
    \end{align*}
    So \cref{eq:coststay-concentration} implies \cref{eq:coststay-concentration-error}. 
    Similarly,
    \begin{align*}
        |\rmEstCostMoves_w(S,s,v) - \rmEstCostMoves_w(S,\tilde{s},v)| & = \left|\frac{s - \tilde{s}}{\eta_0}\sum_{i = 1}^{\eta_0}(d_w(v, \{u_i\}) + d(v, \{u_i\}))\right| \\
        & \leq \frac{|s - \tilde{s}|}{\eta_0}\sum_{i = 1}^{\eta_0}|d_w(v, \{u_i\}) + d(v, \{u_i\})| \\
        & \leq \frac{\epsilon's}{\eta_0}\eta_0(W+1) \\
        & < 2\epsilon'Ws
    \end{align*}
    So \cref{eq:costmove-concentration} implies \cref{eq:costmove-concentration-error}. 
    Therefore this lemma holds from lemma~\ref{lem:concentration-weighted}. 
\end{proof}

\begin{corollary}
    \label{cor:concentration}
    In the setting of lemma~\ref{lem:concentration-weighted-noisy}, let $\epsilon' = \frac{1}{\eta^2}$, with probability at least $1 - 4\exp\left(-\frac{1}{2}\eta\right)$ we have that the following
two inequalities hold:
\begin{equation}
    \label{eq:cor-coststay-concentration-error}
    \rmEstCostStays_w(S,\tilde{s},v) \in \rmExpCostStays_w(s,v) \pm \frac{3}{\eta^2}Ws
\end{equation}
\begin{equation}
    \label{eq:cor-costmove-concentration-error}
    \rmEstCostMoves_w(S,\tilde{s},v) \in \rmExpCostMoves_w(s,v) \pm \frac{3}{\eta^2}Ws
\end{equation}
\end{corollary}

In the remainder of this section, we fix $\epsilon' = \frac{1}{\eta^2}$ for simplicity. 

\paragraph{The algorithm.}
The algorithm maintains a tentative solution $\calC$ and iteratively tries to find
an improving cluster.
The algorithm is parameterized by an integer parameter $\eta > 0$, which we will fix later.
In one step, the algorithm invokes the following \textbf{GenerateCluster} routine
for any choice of $r \in V(G)$,  
multisets $(S^i)_{i=1}^\eta$ of vertices in $N(r)$ of size $\eta^5$ each,
and integers $(s^i)_{i=1}^\eta \in [|V|]^\eta$.
If any of the run finds a cluster $S'$ such that $\cost_w(\calC + S') < \cost_w(\calC)$,
it picks $S'$ that minimizes $\cost_w(\calC + S')$, replaces $\calC := \calC+S'$ and restarts.
If no such $S'$ is found, the algorithm terminates and returns the current solution $\calC$
as the final output.
\begin{tcolorbox}[beforeafter skip=10pt]
    \textbf{Polynomial-Time-Local-Search($\eta$)}
    \begin{itemize}[itemsep=0pt, topsep=4pt]
    \item Compute an atomic pre-clustering $(\calK, E^{\adm})$ using the Precluster algorithm.
    \item $\calC \gets \{K \mid K \in \calK\} \cup \{\{v\} \mid v \in V \setminus (\cup_{K \in \calK}K)\}$.
    \item Do
    \begin{itemize}[itemsep=0pt, topsep=4pt]
        \item For each vertex $r$, for each collection of  $\eta$ (not necessarily disjoint) multisets $\mathcal{S} = S^1,\ldots,S^\eta$ each of $\eta_0=\eta^{5}$ vertices in $N(r)$, and each collection of $\eta$ integers $\vec{s} = s^1,\ldots,s^\eta \in [|V|]$, \\
        $\calC(\mathcal{S},\vec{s}) \gets  \calC + $ GenerateCluster($\calC, r, S^1,\ldots,S^\eta, s^1,\ldots,s^\eta$).
        \item Let $\mathcal{S}^*, \vec{s}^*$ be the pair such that the cost of $\calC(\mathcal{S}^*, \vec{s}^*)$ is minimized.
        \item If the $\cost(\calC + \mathcal{S}^*) < \cost(\calC)$:
        \begin{itemize}
            \item has\_improved $\gets$ true;
            \item $\calC \gets \calC(\mathcal{S}^*, \vec{s}^*)$;
        \end{itemize}
        \item Else has\_improved $\gets$ false
        \end{itemize}
        \item While has\_improved
        \item Output $\calC$
    \end{itemize}
\end{tcolorbox}

\begin{tcolorbox}[beforeafter skip=10pt]
    \textbf{GenerateCluster($\calC, r, S^1,\ldots,S^\eta, s^1,\ldots, s^\eta$)}    
\begin{itemize}[itemsep=0pt, topsep=4pt]
    \item $S' \gets K(r)$.
    \item Let $D_r^1, \ldots, D_r^{\eta}$ be an arbitrary partition of the
    vertices of $D(r)$ into equal-size parts,
    \item For $i=1,\ldots,\eta$:
    \begin{itemize}
    \item For each vertex $v \in D_r^i$, 
    \begin{itemize}[itemsep=0pt, topsep=4pt]
    \item Let $C(v)$ be the current cluster of $v$. 
    \item If EstCostStays($S^i,s_i,v$) $>$ EstCostMoves($S^i, s_i,v$) + $6\eta^{-1} W|N(r)|$,\\ then $S' \gets S' \cup \{v\}$.
    \end{itemize}
    \end{itemize}
    \item Return $S'$.
\end{itemize}
\end{tcolorbox}

Note that in \cref{alg:iterative-local-search}, we only use $\{1, 1.5, 2\}$  in the weight, and therefore we only have polynomial many possible cost for a clustering scheme. 
Therefore, Polynomial-Time-Local-Search runs in polynomial iterations and takes polynomial time, as $\eta, \epsilon$ be some constants. 
In the next sections, we will remove this assumption and consider arbitrary bounded weight function. 
And the following lemma shows that for sufficiently large integer $\eta$, the above algorithm
indeed finds a $\gamma$-good local optimum.
Note that for the polynomial-time implementation, we would only need a weaker statement saying that there exists a choice of $(S^i)_{i=1}^\eta$
and $(\tilde{s}^i)_{i=1}^\eta$ for which the statement holds, as we iterate
over all possible choices. However, the fact that a vast majority of choices
leads generating an improving cluster $S'$ is crucial for subsequent running
time improvements, where we will only sample sets $S^i$ and integers $\tilde{s}^i$.

\begin{lemma}
\label{lem:generate-cluster-4.2}
Let $0 < \gamma < \varepsilon^{13}/4$ and $\eta > 4608W \varepsilon^{-8} \gamma^{-1}$ be a sufficiently large constant so that
$\exp(\eta) > 400\eta\varepsilon^{-6}$.

Consider a clustering $\calC$ maintained by the local search algorithm and assume that there exists a vertex $r$, and a $(\varepsilon, \varepsilon/2)$-good cluster $C$ with $K(r) \subseteq C \subseteq N(r)$
such that $|C| > 1$ and 
\[
    \cost_w(\calC + C)  + \gamma \sum_{u \in C} d^{\adm}(u) \le \cost_w(\calC) ~.
\]
Then there exists a collection of sets $C^*_1, \dots, C^*_{\eta}$ of vertices
with the following properties.
First, $K(r) \subseteq C^*_i \subseteq N(r)$, $|C^*_i| \geq \frac{1}{112896W} \gamma\varepsilon^{14}|N(r)|$. Second, let $\mathcal{S} = S^1,\ldots,S^\eta$, where $S^i$ is an uniform sample of $C^*_i$ of size $\eta^5$, let $\vec{s} = (\tilde{s}^1 ,\ldots, \tilde{s}^\eta) \in [|V|]^{\eta^5}$ such that $\tilde{s} \in (1 \pm \epsilon')|C^*_i|$. Then with probability at least $1 - 2\eta\exp(-\eta)$, 
the cluster $S'$ output by GenerateCluster($\calC, r, K(r), \mathcal{S}, \vec{s})$ satisfies
\[\cost_w(\calC+S') \leq \cost_w(\calC + C) +  \frac{\gamma}{2}\sum_{u \in C}d^{\adm}(u).\]
\end{lemma}
\begin{proof}
Fix $r$ as in the lemma statement and let 
\[ C^* = \text{argmin}_{K(r) \subseteq Q \subseteq N(r)} \cost_w(\calC+Q).\]
That is,  $C^\ast$ is the most
improving cluster in $N(r)$ containing $K(r)$. 
Since $C$ is $(\varepsilon, \varepsilon/2)$-good cluster containing $r$, and that each $(\varepsilon, \varepsilon/2)$-good cluster contains at most one atom, we have that $C \subseteq N(r)$. 
Thus, we have that $\cost_w(\calC+C^*) \le \cost_w(\calC+C) \le \cost_w(\calC)
- \gamma \sum_{u \in C} d^{\adm}(u)$, so it is enough to show that GenerateCluster outputs a cluster of cost at most $\cost(\calC+C^*) + \frac{\gamma}{2} \sum_{u \in C} d^{\adm}(u)$.

Since $C$ is a $(\varepsilon,\varepsilon/2)$-good cluster containing $K(r)$, $|C| \geq \varepsilon d(v)$ for all $v \in K(r)$, therefore $|N(r)| \geq |C| \geq \frac{1}{2}\varepsilon d(v)$. 

\begin{claim}
    \label{claim:degree-4.2}
    For any $v \in N(r)$,
    \[ d(v) \leq 4\varepsilon^{-2}|N(r)|\quad\mathrm{and}\quad d^{\adm}(v) \leq 8\varepsilon^{-5}|N(r)|.\]
\end{claim}
\begin{proof}
We first prove the first inequality.

If $r \in K(r)$, then $|N(r)| \geq |C| \geq \frac{1}{2}\varepsilon d(v)$; otherwise, $v$ is admissible to $r$, therefore $d(v) \leq 2\varepsilon^{-1}d(r) \leq 4\varepsilon^{-2}|N(r)|$. 

The second inequality follows from the first and the condition $d^{\adm}(v) \leq 2\varepsilon^{-3}d(v)$.
\end{proof}

Let $s^* = |C^*|$. Next, we will prove a lower bound on $s^*$. 
To this end, we show a quite brute-force estimate that adding one new cluster to a clustering scheme cannot improve the cost
by too much.

\begin{claim}\label{claim:single-improvement-4.2}
    Let $\calC$ be a clustering maintained by the local search algorithm, $C \subseteq |N(r)|$ is any cluster, then
    \[ \cost_w(\calC) - 49W\varepsilon^{-6}|C|\cdot |N(r)| \leq \cost_w(\calC + C). \]
\end{claim}
\begin{proof}
    We split the cost improvement into two terms. The first term is the improvement on $E^+$, which is at most $W \cdot \frac{1}{2}|C|(|C| - 1) \leq W \cdot |C| \cdot |N(r)|$. 
    The second term is the improvement on $E^-$, which is at most the maximum amount of cost we paid before. 
    Let $v$ be an arbitrary vertex in $C$, $C_i$ be the cluster in $\calC$ containing $v$, the minus cost we pay in $\calC$ at $v$ is at most $|C_i| - 1$. 
    According to Corollary~\ref{cor:bounded-minus-cost}, $|C_i| \leq 12\varepsilon^{-4}d(v)$. 
    According to claim~\ref{claim:degree-4.2}, for all $v \in N(r)$, $d(v) \leq 4\varepsilon^{-2}|N(r)|$. 
    So we have the total improvement on $E^-$ is at most $12\varepsilon^{-4}|C| \cdot (4\varepsilon^{-2}|N(r)|) = 48\varepsilon^{-6}|C|\cdot |N(r)|$. 
    Together with the improvement in the first part, the whole claim holds. 
\end{proof}

Observe that by Lemma~\ref{lem:D-times-N-4.2}, we have
\[ \sum_{u \in C} d^{\adm}(u) \geq \frac{1}{576}\varepsilon^8|D(r)|\cdot|N(r)|\]
Hence, 
\[\cost_w(\calC + C^*) \leq \cost_w(\calC) - \gamma \sum_{u \in C}d^{\adm}(u) \leq \cost(\calC) - \frac{1}{576}\gamma \varepsilon^8 |D(r)|\cdot|N(r)|.\]
Together with Claim~\ref{claim:single-improvement-4.2} for $C^*$,
it follows that
\begin{equation}\label{eq:s-ast-bound-4.2}
s^* = |C^*| \geq \frac{1}{28224W} \gamma \varepsilon^{14} |D(r)|.
\end{equation}

Let $v$ be a vertex in $D(r)$, according to the optimality of $C^*$, we have
\begin{enumerate}
    \item if $v$ in $C^*$, then $\cost_w(\calC+(C^*\setminus \{v\})) \ge \cost_w(\calC+C^*)$; and
    \item if $v$ is not in $C^*$ then $\cost_w(\calC+(C^* \cup \{v\})) \ge 
    \cost_w(\calC+C^*)$.
\end{enumerate}

To analyze the algorithm, let $S'^{1},\ldots,S'^{\eta}$ be the set of elements in set $S'$ 
at the beginning of the $i$th iteration of the for loop. 
We define 
\[ C^*_i = S'^{i} \cup \mathrm{argmin} \left\{ \cost_w(\calC + (S'^{i} \cup T))~\Big|~T \subseteq \bigcup_{j = i}^{\eta}D_r^i\right\}.\]
Let $S^i$ be a uniform sample of the cluster $C^*_i$ and $s_i = |C^*_i|$. 
We want to show the following claim.

\begin{claim}
\label{claim:invariant-4.2}
With probability at least $1 - 2\exp(-\frac{1}{2}\eta)$,
\[ \cost_w(\calC + C^*_{i+1}) - \cost_w(\calC+C^*_i) \le 4W \eta^{-2} |D(r)| \cdot |N(r)|. \] 
\end{claim}

Assuming the claim is true, with probability at least $1 - 2\eta \exp(-\frac{1}{2}\eta)$, this claim holds for all $i = 1, \dots, \eta$. 
Then the lemma follows by observing that $S' = C^*_{\eta+1}$ and $C^*_1 = C^*$. 
By telescoping and Lemma~\ref{lem:D-times-N-4.2}, we have
\begin{align*}
    \cost_w(\calC+S') & = \cost_w(\calC+C^*) + \sum_{i = 1}^{\eta}(\cost_w(\calC + C^*_{i+1}) - \cost_w(\calC + C^*_i)) \\
    & \le \cost_w(\calC + C^*) + 4W \eta^{-1} \cdot |D(r)|\cdot |N(r)| \\ 
    & \le \cost_w(\calC+C^*) + 4W \eta^{-1} \cdot 576 \varepsilon^{-8} \sum_{u \in C}d^{\adm}(u) \\
    & \le \cost_w(\calC + C)  + \frac{\gamma}{2} \sum_{u \in C}d^{\adm}(u)
\end{align*}
The last inequality follows from the assumption $\eta > 4608W\varepsilon^{-8}\gamma^{-1}$. 

The only thing we miss is to show a lower bound on $|C^*_i|$.
By telescoping, we have
\begin{align*}
    \cost_w(\calC + C^*_i) & = \cost_w(\calC + C^*) + \sum_{j = 1}^{i - 1}(\cost_w(\calC + C^*_{j+1}) - \cost_w(\calC + C^*_i)) \\
    & \leq \cost_w(\calC + C^*) + (i - 1) \cdot 4W\eta^{-2}|D(r)|\cdot|N(r)| \\
    & \leq \cost_w(\calC) - \frac{1}{576}\gamma \varepsilon^8|D(r)|\cdot|N(r)| + 4W\eta^{-1}|D(r)|\cdot|N(r)| \\
    & \leq \cost_w(\calC) - \frac{1}{1152}\gamma \varepsilon^8|D(r)|\cdot|N(r)|
\end{align*}
The last inequality holds from the assumption $\eta > 4608W\varepsilon^{-8}\gamma^{-1}$ and $\varepsilon < 1$. 
Together with claim~\ref{claim:single-improvement-4.2} for $C^*_i$, it follows that $|C^*_i| \geq \frac{1}{56448W} \gamma\varepsilon^{14}|D(r)|$. 
Since $K(r) \subseteq C^*_i$ and $|N(r)| = |D(r)| + |K(r)|$, we have
\begin{equation}\label{eq:size-of-sets-lowerbound-4.2}
|C^*_i| \geq \frac{1}{112896W} \gamma\varepsilon^{14}|N(r)|.
\end{equation}

\begin{proof}[Proof of Claim~\ref{claim:invariant-4.2}]
Let 
\[ Q_i = S'^{i+1} \cup \left( C_i^\ast \cap \bigcup_{j=i+1}^\eta D_r^j\right).\]
By the definition of $C^*_{i+1}$, we have 
\[\cost_w(\calC + C^*_{i+1}) \leq \cost_w(\calC + Q_i),\] so it is sufficient to prove that 
\[\cost_w(\calC + Q_i) - \cost_w(\calC + C^*_i) \leq 4W \eta^{-2} |D(r)| \cdot |N(r)|.\]

We observe that 
\[ \cost_w(\calC + Q_i) - \cost_w(\calC + C^*_i) \leq  W|Q_i \Delta C^*_i|^2 + \sum_{v \in Q_i \Delta C^*_i}\left(\cost_w(\calC + Q_i, v) - \cost_w(\calC + C^*_i, v)\right)~,\]
where $\Delta$ denote the symmetric difference of two sets. 
Fix a vertex $v \in Q_i \Delta C^*_i$ and let $C \in \calC$ be the cluster containing $v$.
We have 
\[\cost_w(\calC + Q_i, v) - \cost_w(\calC + C^*_i, v) \leq W(|C| + |Q_i| + |C^*_i|)~,\]
where $\cost_w(\calC, v)$ is the total cost paid by clustering scheme $\calC$ for all edges or non-edges incident to $v$. 
According to Corollary~\ref{cor:bounded-minus-cost}, $|C| \leq 12\varepsilon^{-4}d(v) \leq 48\varepsilon^{-6}|N(r)|$. 
Since $Q_i, C^*_i$ are subsets of $N(r)$, $|Q_i| \leq |N(r)|$, $|C^*_i| \leq |N(r)|$, therefore 
\[\cost_w(\calC + Q_i, v) - \cost_w(\calC + C^*_i, v) \leq W(|C| + |Q_i| + |C^*_i|) \leq 50\varepsilon^{-6}W|N(r)|.\]

A vertex $v \in Q_i \Delta C^*_i$ can be one of the following three types. 
\begin{itemize}
    \item $v$ is a \emph{missing} node if $v \in Q_i$ and $v \notin C^*_i$. 
    \item $v$ is a \emph{core} node if $v \in C^*_i, v \notin Q_i$ and 
    \[ \rmCostMoves(C^*_i,v) + 12\eta^{-2} W|N(r)| < \rmCostStays(C^*_i,v). \]
    \item $v$ is a \emph{non-core} node if $v \in C^*_i, v \notin Q_i$ and \[ \rmCostMoves(C^*_i,v) + 12\eta^{-2} W|N(r)| \geq \rmCostStays(C^*_i,v). \]
\end{itemize}

In the following part, we will split the cost difference into vertices of the three types, and give a bound separately. 

\textbf{Missing nodes.} If $v \notin C^*_i$, by definition of $C^*_i$, $\rmCostMoves(C^*_i,v) \geq \rmCostStays(C^*_i,v)$. 
According to Corollary~\ref{cor:concentration}, with probability at most 
$4\exp\left(-\frac{1}{2}\eta\right)$, we have
\[\rmEstCostMoves(S^i,\tilde{s}^i,v) + 6\eta^{-2}W|N(r)| < \rmEstCostStays(S^i,\tilde{s}^i,v).\]
Therefore $v \in Q_i$ with probability at most $4\exp\left(-\frac{1}{2}\eta\right)$,
and the expected size of $Q_i \setminus C^*_i$ is
bounded by $4\exp\left(-\frac{1}{2}\eta\right)\eta^{-1}|D(r)|$. 
Hence, there are more than $4\exp\left(-\eta\right)\eta^{-1}|D(r)|$
missing nodes with probability at most $\exp\left(-\frac{1}{2}\eta\right)$. 

\textbf{Core nodes.} If $v \in C^*_i$ is a core node, by definition of core node, 
\[ \rmCostMoves(C^*_i,v) + 12\eta^{-2} W|N(r)| < \rmCostStays(C^*_i,v). \]
According to corollary~\ref{cor:concentration}, we have
\[\rmEstCostMoves(S^i,\tilde{s}^i,v) + 6\eta^{-2} W|N(r)| \geq \rmEstCostStays(S^i,\tilde{s}^i,v)\]
with probability at most $4\exp\left(-\frac{1}{2}\eta\right)$.
Therefore $v \notin Q_i$ with probability at most $4\exp\left(-\frac{1}{2}\eta\right)$. 
Similarly as before, there are more than $4\exp\left(-\eta\right)\eta^{-1}|D(r)|$
core nodes with probability at most $\exp\left(-\frac{1}{2}\eta\right)$.

\textbf{Non-core nodes.} If $v \in C^*_i$ is a non-core node, 
then 
\begin{align*}
 &\cost_w(\calC + Q_i, v) - \cost_w(\calC + C^*_i, v) \\
 &\qquad \leq \rmCostStays(C^*_i, v) - \rmCostMoves(C^*_i, v) + W|Q_i \Delta C^*_i| \\
 &\qquad \leq 12\eta^{-2}W|N(r)| + W\eta^{-1}|D(r)| \\
 & \qquad \leq 2\eta^{-1}W|N(r)|. 
\end{align*}
In the inequalities, we use the fact that
$Q_i \Delta C_i^\ast \subseteq D_i^r$,  $|D_i^r| = \eta^{-1} |D(r)|$ and $\eta < 1/12, W > 1$. 

Summing up the total cost difference over all vertices of the three types, we have
\begin{align*}
    &\cost_w(\calC + Q_i) - \cost_w(\calC + C^*_i) \\
    &\qquad \leq W|Q_i \Delta C_i^\ast|^2 + \sum_{v \in Q_i \Delta C^*_i}\left(\cost_w(\calC + Q_i, v) - \cost_w(\calC + C^*_i, v)\right) \\
    &\qquad \leq W\eta^{-2} |D(r)|^2 + 2 \cdot 4\exp(-\eta)\eta^{-1}|D(r)| \cdot 50\varepsilon^{-6}W|N(r)|
    + \eta^{-1} |D(r)| \cdot 2\eta^{-1}W |N(r)| \\
    &\qquad \leq 4\eta^{-2} W|D(r)| \cdot |N(r)|.
\end{align*}
The last inequality holds since $\exp(\eta) > 400\eta\varepsilon^{-6}$. 
This finishes the proof of the claim.
\end{proof}
This finishes the proof of the lemma.
\end{proof} 

\begin{corollary}
    Under the setting of \cref{lem:generate-cluster-4.2}, 
    \textbf{Polynomial-Time-Local-Search} will output a $\gamma$-good local optimum $\calC$ with probability at least $1 - 2\eta\exp(-\eta)$. 
\end{corollary}

\begin{proof}
    By contradiction, suppose the output $\calC$ from \textbf{Polynomial-Time-Local-Search} is not a $\gamma$-good local optimum, let $\calC^*_{(\calK, E_{\adm})}$ be the (unknown) $\epsilon$-good clustering scheme in $(\calK, E_{\adm})$ from \cref{thm:preprocessed-wrapper}, by definition
    \[
        \sum_{C\in\calC^*_{(\calK,E_{\adm})}} ( \cost_w(\calC) - \cost_w(\calC + C) ) > 2 \gamma |E_{\adm}| ~.    
    \]
    Since $\sum_{C\in\calC^*_{(\calK,E_{\adm})}}\sum_{u \in C}d^{\adm}(u) = 2|E_{\adm}|$, we have
    \[
        \sum_{C\in\calC^*_{(\calK,E_{\adm})}} ( \cost_w(\calC) - \cost_w(\calC + C) - \gamma\sum_{u \in C}d^{\adm}(u)) > 0 ~.
    \]
    According to pigeonhole principle, there exists a cluster $C^*_i \in \calC^*_{(\calK, E_{\adm})}$ such that
    \[
        \cost_w(\calC) - \cost_w(\calC + C^*_i) - \gamma\sum_{u \in C^*_i}d^{\adm}(u) > 0
    \]
    According to \cref{lem:generate-cluster-4.2}, in this case, with probability at least $1 - 2\eta\exp(-\eta)$, in the last iteration, the cluster $S'$ output by GenerateCluster($\calC, r, K(r), \mathcal{S}, \vec{s})$ in the last iteration satisfies
    \[
    \cost_w(\calC+S') \leq \cost_w(\calC + C^*_i) + \frac{\gamma}{2}\sum_{u \in C^*_i}d^{\adm}(u) < \cost_w(\calC) - \frac{\gamma}{2}\sum_{u \in C^*_i}d^{\adm}(u) ~,
    \]
    so the iteration will continue, which lead to a contradiction. 
\end{proof}

\section{Faster Implementation of Local Search}

In this section, we show a faster implementation of the local search algorithm, by using sampling to find new improving clusters instead of enumerating a lot. This is not the end, but both the algorithm and the analysis are important preparations for our following improvements. We define $d^{adm}(C)=\sum_{u\in C}d^{adm}(u)$. We will consider the following local search algorithm, with parameters $\eta,\gamma$ and $s$. 
\begin{tcolorbox}[beforeafter skip=10pt]
    \label{alg:near-linear}
    \textbf{Faster-Local-Search($\eta,\gamma,s$)}
    \begin{itemize}[itemsep=0pt, topsep=4pt]
    \item Compute an atomic pre-clustering $(\calK, E^{\adm})$ using the Precluster algorithm.
    \item $\calC \gets \{K \mid K \in \calK\} \cup \{\{v\} \mid v \in V \setminus (\cup_{K \in \calK}K)\}$.
    \item Do many rounds of the following until $\calC$ does not change in $\Theta(n\log n)$ consecutive rounds:
        \begin{itemize}[itemsep=0pt, topsep=4pt]
        \item Sample a vertex $r'$ with probability $\frac1n$ and try to improve the clustering by the singleton $\{r'\}$.
        \item Sample a vertex $r$ with probability $\frac{1}{n\cdot d(r)}$, or instantly finish this round with probability $1-\sum_{r\in V}\frac{1}{n\cdot d(r)}$.
        \item Uniformly sample $\eta$ subsets $T^1,\dots,T^\eta$ of size $s$ from $N(r)$.
        \item For each collection of $\eta$ (not necessarily disjoint) multisets $\mathcal{S} = S^1,\ldots,S^\eta$ each of $\eta^{5}$ vertices in $T^i$ respectively, and each collection of $\eta$ integers $\vec{s} = s^1,\ldots,s^\eta \in \{\epsilon d(r) (1+\epsilon')^k/2 | 0\le k\le \log_{1+\epsilon'}(4/\epsilon)\}$, \\
        $\calC(\mathcal{S},\vec{s}) \gets  \calC + $ GenerateCluster($\calC, r, S^1,\ldots,S^\eta, s^1,\ldots,s^\eta$). 
        \item Let $\mathcal{S}^*, \vec{s}^*$ be the pair such that the cost of $\calC(\mathcal{S}^*, \vec{s}^*)$ is minimized.
        \item If the cost of $\calC(\mathcal{S}^*, \vec{s}^*)$ is at least $\gamma|E_{\adm}|/n$ less than the cost of $\calC$:
            \begin{itemize}
                \item $\calC \gets \calC(\mathcal{S}^*, \vec{s}^*)$;
            \end{itemize}
        \end{itemize}
    \item Output $\calC$
    \end{itemize}
\end{tcolorbox}

\paragraph{The size of the sampled set $S$.}

For GenerateCluster to work, we need $S^i$ to be a uniform sample of $C_i^*$. When $|T^i\cap C^*_i|\ge\eta^5$ for all $i$, we have at least one valid sample $\mathcal{S}$. Since each $T^i$ is uniformly sampled from a superset of $C_i^*$, the samples we get are uniform.

\begin{lemma}\label{lem:sampled-enough}
    If $s>10^6\eta^6\epsilon^{-27}$, then $|S\cap C^*_i|\ge\eta^5$ for all $i$ with probability at least $\frac12$.
\end{lemma}

\begin{proof}
    By Lemma \ref{lem:generate-cluster-4.2}, $|C^*_i|\ge \frac{\epsilon^{27}}{451584}|N(r)|$. It follows by a Chernoff bound for each $i$ and then a union bound over $i$.
\end{proof}

\paragraph{The number of rounds.} We know that the cost of our clustering will be improved by at least $\gamma|E_{\adm}|/n$ per $\Theta(n\log n)$ rounds. On the other hand, the following lemma provides an upper bound of the initial cost.

\begin{lemma}
\label{lem:initial-cost}
    The cost of our initial solution which consists of atoms and singletons is at most $\cost(\opt)+\frac{4}{\epsilon}|E_{adm}|$.
\end{lemma}
    
\begin{proof}
    First, the minus cost paid by our initial solution comes from atoms, so it is also paid by $\opt$. The plus cost we paid induced by edges between atoms, or incident to a singleton in $\opt$, is again also paid by $\opt$. So our extra cost is the number of edges incident to vertices which does not belong to any atom and is not singleton in $\opt$. For any such vertex $v$, since there is a good cluster containing it, $d^{adm}(v)+1\ge\epsilon d(v)$. Hence the additional cost is at most $\sum_{v}\frac{d^{adm}(v)+1}{\epsilon}\le \sum_{v}\frac{2d^{adm}(v)}{\epsilon}\le\frac{4}{\epsilon}|E_{adm}|$.
\end{proof}

So, taking e.g. $\gamma=\epsilon^{20}$, the algorithm stops in $O_{\epsilon}(n^2\log n)$ rounds.

\paragraph{Time complexity.}

The number of rounds is $O_{\epsilon}(n^2\log n)$. In each round, each vertex $r$ is sampled with probability $\frac{1}{n\cdot d(r)}$. If $r$ is sampled, in each time calling GenerateCluster, we need to go through its neighbors and for each neighbor $u$, it takes $O(d(u))$ time to estimate the costs of moving and staying, so it takes $O_{\epsilon}(d(r)^2)$ time in total. There are only constant many $\mathcal{S}$'s and $\vec{s}$'s. So the total time complexity is $O_{\epsilon}(n^2\log n \sum_r \frac{1}{n\cdot d(r)} d(r)^2)=O_{\epsilon}(mn\log n)$.

\paragraph{Correctness.}
\label{sec:number-of-rounds}

In each round, we sample each vertex $r$ with probability $\frac{1}{n\cdot d(r)}$. Suppose the optimal good clustering is $\opt=\{C_1,\dots,C_k\}$. For each cluster $C_i$, if it contains an atom $K$ such that $\frac{|K|}{|C_i|}<\frac{\epsilon^{21}}{576}$, we split the cluster into $K$ and $C_i\setminus K$. Let $\opt'$ be the new clustering we get. By the following lemmas, it suffices to compare our solution with $\opt'$ instead of $\opt$.

\begin{lemma}
    $\opt'$ is $(\epsilon,\frac{\epsilon}{2})$-good.
\end{lemma}

\begin{proof}
    Consider any cluster $C_i\in\opt$. If $C_i$ is not split, then it's $(\epsilon,\epsilon)$-good and hence $(\epsilon,\frac{\epsilon}{2})$-good. Otherwise it is split into an atom $K$ and $C_i\setminus K$, where $K$ is $(\epsilon,\frac{\epsilon}{2})$-good since it's an atom, and $C_i\setminus K$ is $(\epsilon,\frac{\epsilon}{2})$-good because $|C_i\setminus K|\ge \frac{|C_i|}{2}$.
\end{proof}

\begin{lemma}
\label{lem:cost-opt'}
    $\cost(\opt') \le (1+\epsilon)\cost(\opt)$.
\end{lemma}

\begin{proof}
    Consider any cluster $C_i\in\opt$. If we split $C_i$ into $K$ and $C_i\setminus K$, we can only increase the cost by at most $|K|\cdot|C_i\setminus K| \le \frac{\epsilon^{21}}{576}|C_i|\cdot|C_i\setminus K| \le \frac{\epsilon^{21}}{576}|N(K)|\cdot|D(K)|$ since $C_i\subseteq N(K)$. Then by Lemma \ref{lem:D-times-N-4.2}, this is at most $\frac{\epsilon^{13}}{4}d^{adm}(C_i)$. Summing over all clusters, the lemma follows from Theorem \ref{thm:preprocessed-wrapper}.
\end{proof}

Consider any cluster $C_i'\in \opt'$ of size larger than $1$. If it contains an atom $K$, we will hit $K$ with probability at least $\frac{\epsilon|K|}{n|C_i'|} \ge \frac{\epsilon^{22}}{576n}$; otherwise we will hit some vertex in $C_i'$ with probability at least $\frac{\epsilon}{n}$. For $|C_i'|=1$, we will hit it with probability $\frac1n$.

Define $\Delta_i=\max\{\cost(\mathcal{C})-\cost(\mathcal{C}+C_i')-\gamma d^{adm}(C_i'),0\}$ for $0<\gamma<\frac{\epsilon^{13}}{4}$, and $\Delta=\sum_i\Delta_i$. If we hit $C_i'$, by Lemma \ref{lem:generate-cluster-4.2}, we can improve the cost by $\Delta_i$ with constant probability. So, if some $\Delta_i$'s are at least $\gamma|E_{adm}|/n$, then in the following $\Theta_{\epsilon}(n\log n)$ rounds, with high probability we can observe one of them and improve $\calC$. Hence with high probability we won't stop with some $\Delta_i\ge\gamma|E_{adm}|/n$ at any round $t$. By an union bound over (a polynomial number of) rounds, we know that with high probability, we will get a clustering with $\Delta\le\gamma|E_{adm}|$.

\begin{lemma}
    When $\Delta\le\gamma|E_{adm}|$, the clustering we have is a $\frac32\gamma$-good local optimum (with respect to $\opt'$).
\end{lemma}

\begin{proof}
    By the definition of $\Delta$, $\sum_i(\cost(\mathcal{C})-\cost(\mathcal{C}+C_i'))\le\sum_i(\Delta_i+\gamma d^{adm}(C_i'))=\Delta+2\gamma|E_{adm}|$. The lemma simply follows.
\end{proof}

We summarize our conclusion in the following theorem.

\begin{theorem}\label{thm:faster-gamma-optimal}
    With high probability, the clustering $\mathcal{C}$ returned by our faster implementation is a $\frac32\gamma$-good local optimum.
\end{theorem}

\paragraph{Further improvements.}

In the remaining sections, we show how to implement our local search algorithm in MPC, sublinear and streaming models. In each round, our algorithm can be viewed as two steps:
\begin{itemize}
    \item Randomly sample at most one pivot vertex;
    \item Try to improve the current clustering by sampling a constant number of (possibly dependent) clusters containing the pivot.
\end{itemize}

In Section \ref{sec:MPC}, we show how to sample multiple pivots in one round, reducing the number of rounds to a constant, while a similar result with Lemma \ref{lem:generate-cluster-4.2} still holds.

In Sections \ref{sec:sublinear}, we show another way of sampling clusters for each pivot and estimating the costs, which work in the sublinear and streaming models.

\section{MPC implementation of local search}
\label{sec:MPC}
In this section we will show how to implement the local search within constant number of rounds in the MPC model. Note that we are not aware of a way to implement the pivot step in the algorithm in \cref{alg:iterative-local-search} efficiently under the MPC model, but we can implement the algorithm in \cref{alg:local-search-flip} and achieves a $2-1/8+\epsilon$ approximation as long as we can implement the local search efficiently. 

\begin{tcolorbox}[beforeafter skip=10pt]
    \textbf{MPC-Local-Search($\eta,\gamma,\gamma', s$)}
    \begin{itemize}[itemsep=0pt, topsep=4pt]
    \item Compute an atomic pre-clustering using the Precluster algorithm.
    \item $\calC \gets \{K \mid K \in \calK\} \cup \{\{v\} \mid v \in V \setminus (\cup_{K \in \calK}K)\}$.
    \item While $\calC$ improves: %
        \begin{itemize}[itemsep=0pt, topsep=0pt]
            \item Do the following $\Theta_\epsilon(\log n)$ times in parallel:
            \begin{itemize}
                \item For each vertex, try to improve the clustering by the singleton.
                \item Select each vertex $r$ with probability $\frac{\epsilon^{4}\gamma'}{24d(r)}$. For each selected vertex:
                \begin{itemize}
                    \item Mark all its neighbors.
                    \item Check the proportion of its neighbors that are marked by other selected vertices. Stop if this proportion is at least $\gamma'$.
                \end{itemize}
                \item Uniformly sample $\eta$ subsets $T^1,\dots,T^\eta$ of size $s$ from $N(r)$.
                \item For each collection of $\eta$ (not necessarily disjoint) multisets $\mathcal{S} = S^1,\ldots,S^\eta$ each of at most $\eta^{5}$ vertices in $T^i$ respectively, and each collection of $\eta$ integers $\vec{s} = s^1,\ldots,s^\eta \in \{\epsilon d(r) (1+\epsilon')^k/2 \mid 0\le k\le \log_{1+\epsilon'}(4/\epsilon)\}$, \\
                $f(\mathcal{S},\vec{s}) \gets $ GenerateCluster($\calC, r, S^1,\ldots,S^\eta, s^1,\ldots,s^\eta$). 
                \item Find the $(\mathcal{S}^*,\vec{s}^*)$ that maximises the improvement $f(\mathcal{S},\vec{s})$.
                \item If the improvement is at least $\gamma|E_{\adm}|$:  
                \begin{itemize}
                    \item Update $\calC$ by all the new found clusters using $(\mathcal{S}^*,\vec{s}^*)$ in the same way as we generate the clusters.
                \end{itemize}
            \end{itemize}
            \item Choose the best of the $\Theta(\log n)$ improvements to improve $\calC$
        \end{itemize}
    \item Output $\calC$
    \end{itemize}
\end{tcolorbox}

Now GenerateCluster only returns the improvement.

\begin{tcolorbox}[beforeafter skip=10pt]
    \textbf{GenerateCluster($\calC, r, S^1,\ldots,S^\eta, s^1,\ldots, s^\eta$)}    
    \begin{itemize}[itemsep=0pt, topsep=4pt]
    \item Let $D_r^1, \ldots, D_r^{\eta}$ be an arbitrary partition of the
    vertices of $D(r)$ into equal-size parts,
    \item For $i=1,\ldots,\eta$, send $(S^i,s^i)$ to every vertex $v\in D_r^i$.
    \item For each vertex $v \in D_r^i$ which is not double-marked:
        \begin{itemize}[itemsep=0pt, topsep=0pt]
        \item Let $C(v)$ be the current cluster of $v$. 
        \item Compute EstCostStays($S^i,s_i,v$) and EstCostMoves($S^i, s_i,v$) parallel from the neighbors of $v$.
        \item If EstCostStays($S^i,s_i,v$) $>$ EstCostMoves($S^i, s_i,v$) + $6\eta^{-1} W|N(r)|$,\\ decide moving to the new cluster.
        \end{itemize}
    \item For each vertex $v \in D_r^i$ which decides to move:
        \begin{itemize}[itemsep=0pt, topsep=0pt]
        \item Compute the individual improvement parallel from the neighbors of $v$.
        \end{itemize}
    \item Collect the individual improvements and return.
    \end{itemize}
\end{tcolorbox}

Previously when sampling, we have sampled each vertex $r$ with probability of $1/nd(r)$. We will instead sample each vertex $r$ with probability of $\frac{\epsilon^{4}\gamma'}{24d(r)}$. The idea is that this ensures we hit each cluster with constant probability, while still not having too much overlap with other vertices that was selected.

\begin{lemma}\label{lem:mpc-cluster-lost}
    Let $G$ be a graph with an $\epsilon$-good preclustered instance $(\mathcal{K},E_{adm})$. Let each vertex be selected with probability $\epsilon^4\gamma'/24d(r)$. Let $X_u$ for each $u\in V(G)$ be the number of neighbours of $u$ in $D(u)$, that has a neighbour in $E_{adm}$ that was selected. Then for all $u\in V(G)$
    $$\Pr[X_u\ge \gamma'|D(u)|] \le \frac{1}{2}$$
\end{lemma}
\begin{proof}
    First, we calculate the expected value of $X_u$:
    \begin{align*}
        E[X_u] & = \sum_{v\in D(u)}\sum_{w\in N_v}\frac{\epsilon^4\gamma'}{24 d(w)} \\
        & \le \sum_{v\in D(u)}\sum_{w \in N_v}\frac{\epsilon^4\gamma'}{12\epsilon d(v)} \\
        & \le \sum_{v\in D(u)}|N_v|\frac{\epsilon^4\gamma'}{12\epsilon d(v)} \\
        & \le \sum_{v\in D(u)}6\epsilon^{-3}d(v)\frac{\epsilon^4\gamma'}{12\epsilon d(v)} \\
        & \le \frac{\gamma'|D(u)|}{2}
    \end{align*}
    by a union bound and degree similarity of neighbours in $E_{adm}$. Then the lemma follows from Markov's inequality.
\end{proof}
The result of this lemma is that we are only going to stop due to having a too large amount of the neighbours selected with probability at most $1/2$.

To argue that this algorithm works in MPC model, we first note that each vertex only needs to broadcast information to its neighbors, or aggregate the information from its neighbors. When each machine have memory $O(m^{\delta})$, we can have $O(m^{1-\delta}+n)$ machines each taking care of $O(m^{\delta})$ edges incident to the same vertex. For each vertex with degree larger than $m^{\delta}$, we can build a $O(\frac{1}{\delta})$-level B-tree with fan-out $O(m^{\delta})$ to connect all machines related to this vertex. Then the broadcasting and aggregating can be done in $O(\frac{1}{\delta})$ rounds.

\begin{lemma}\label{lem:mpc-success}
    Let $\mathcal{C}$ be a clustering, let $C$ be the optimal improvement for the clustering, let $S'$ be the improving cluster computed by the near linear time local search with $r$ as the starting vertex. Let $S_{mpc}'$ be the improving cluster computed by the MPC local search with $r$ as the selected vertex. Let $\gamma' \le \frac{1}{2304}\varepsilon^8\gamma$ and $\eta$ and $\gamma$ as defined in \cref{lem:generate-cluster-4.2}. Then with probability at least $\frac{1}{4}-2\eta\exp(-\eta)$  
    $$\cost(\mathcal{C} + S_{mpc}') \le \cost(\mathcal{C} + S') + \frac{\gamma}{4}\sum_{u\in C}d^{adm}(u)$$
\end{lemma}
\begin{proof}
    It is clear that the setting of this algorithm is the same as in \cref{lem:generate-cluster-4.2}, except that we have possibly lost a $\gamma'$ fraction of the optimal cluster. We therefore compare our solution to the optimal one selected in \cref{lem:generate-cluster-4.2}. We only continue if we have lost less than a $\gamma'$ fraction of the neighbors.
    Using this we can determine how much additional cost this introduces. 
    \begin{align*}
        \cost(\mathcal{C} + S_{mpc}')-\cost(\mathcal{C} + S') &\le |S_{mpc}'||S'\setminus S_{mpc}'|\\
        &\le \gamma'|N(r)||D(r)|
    \end{align*}
    using the fact that $S_{mpc}'$ is constructed from the same initial vertex $r$ and so shares the same atoms as $S'$. Furthermore, $S_{mpc}'$ is contained in $N(r)$ while $S'\setminus S_{mpc}'$ is contained in $D(r)$. Applying \cref{lem:D-times-N-4.2}, we get that
    $$\cost(\mathcal{C} + S_{mpc}')-\cost(\mathcal{C} + S) \le 576\varepsilon^{-8}\gamma'\sum_{u \in C} d^{adm}(u) = \frac{\gamma}{4}\sum_{u \in C} d^{adm}(u),$$
    by the definition of $\gamma'$.

    With regards to the probability of this happening, observe that the event of getting enough samples, as described in \cref{lem:sampled-enough} and \cref{lem:mpc-cluster-lost} are independent and both with probability $1/2$. By a union bound with the probability of success from \cref{lem:generate-cluster-4.2} we get the probability of success being $\frac{1}{4}-2\eta\exp(-\eta)$.
\end{proof}
From this we see that each time we select a vertex to construct a cluster from, with constant probability we are going to get a constant fraction of the possible improvement for each cluster. Since we can hit a $\Omega_\epsilon(1)$ fraction (instead of $\Omega_\epsilon(\frac{1}{n})$ in previous sections) of the clusters in $\opt'$ in one round, a single round will in expectation achieves a constant fraction of the improvement. This means that between the $\Theta(\log n)$ parallel executions, at least one will with high probability achieve the improvement if it is possible. 

\begin{theorem} %
    When $\Delta \le \gamma|E_{\adm}|$, the clustering $\mathcal{C}$ returned by the MPC implementation of local search is a $\frac{3}{2}\gamma$-good local optimum with high probability.
\end{theorem}
\begin{proof}
Since the setting is essentially the same as in \cref{thm:faster-gamma-optimal}, the proof is the same too.    
\end{proof}

\section{Sublinear and Streaming Implementations of Local Search}
\label{sec:sublinear}
\newcommand{\rmEstImprovement}{\mathrm{EstImprovement}}

Recall Lemmas \ref{lem:sample-neighbor}, \ref{lem:sample-vertex} and Theorem \ref{thm:admissible-query}, which allow us to sample neighbors, (globally) sample vertices, query admissible pairs, and traverse admissible neighbors, in both of the sublinear and streaming models.

We start from our MPC implementation which only has $O(\log n)$ parallel runs in constant rounds, and in each run, we may create multiple disjoint new clusters from multiple selected ``pivots''. There are two operations we need to change to make the algorithm work in sublinear and streaming models: 
\begin{itemize}
    \item When we estimate CostStays and CostMoves, we cannot enumerate all neighbors of $v$. Instead, we need to sample a constant-size set of its neighbors, and use one more Chernoff bound in Lemma \ref{lem:concentration-weighted-noisy};
    \item For each cluster outputted by GenerateCluster, we need to estimate the improvement, and then choose the most (estimated) improving one. We need a concentration bound for the improvement of each round in Section \ref{sec:number-of-rounds}.
\end{itemize}

\paragraph{Estimating CostStays and CostMoves.}
Here we rewrite the subroutine that estimates CostStays.
Let $N_G(v) = \{u \in V \mid (u, v) \in E^+\}$ be the neighborhood of $v$ in the original graph $G$. 

\begin{tcolorbox}[beforeafter skip=10pt]
    \textbf{$\rmEstCostStays'_w(S = \{u_1, \dots, u_{\eta_0}\}, s, v)$}
    \begin{itemize}[itemsep=0pt, topsep=4pt]
        \item Draw $\eta'$ i.i.d uniformly random samples $x_1, \dots, x_{\eta'}$ from $N_G(v)$. 
        \item return $d_w(v) - \frac{|N_G(v)|}{\eta'}\sum_{i = 1}^{\eta'}[x_i \in C(v)]w(x_i, v) - 1 + \frac{s}{\eta_0}\sum_{i=1}^{\eta_0}(d_w(v, \{u_i\}) + d(v, \{u_i\}) - 1)$
    \end{itemize}
\end{tcolorbox}

Recall the subroutine that estimates CostMoves.

\begin{tcolorbox}[beforeafter skip=10pt]
    \textbf{$\rmEstCostMoves_w(S = \{u_1, \dots, u_{\eta_0}\}, s, v)$}
    \begin{itemize}[itemsep=0pt, topsep=4pt]
        \item return $d_w(v) + s - 1 - \frac{s}{\eta_0}\sum_{i=1}^{\eta_0}(d_w(v, \{u_i\}) + d(v, \{u_i\}))$
    \end{itemize}
\end{tcolorbox}

In our algorithm, we only care about the difference between $\rmEstCostStays$ and $\rmEstCostMoves$, so instead of computing them individually, we can implement the following function that compute the difference between the two, to avoid computing $d_w(v)$. 
\begin{tcolorbox}[beforeafter skip=10pt]
    \textbf{$\mathrm{EstCostDiff}'_w(S = \{u_1, \dots, u_{\eta_0}\}, s, v)$}
    \begin{itemize}[itemsep=0pt, topsep=4pt]
        \item $\rmCostStays = s - 1 - \frac{s}{\eta_0}\sum_{i=1}^{\eta_0}(d_w(v, \{u_i\}) + d(v, \{u_i\}))$
        \item Draw $\eta'$ i.i.d uniformly random samples $x_1, \dots, x_{\eta'}$ from $N_G(v)$. 
        \item $\rmCostMoves = -\frac{|N_G(v)|}{\eta'}\sum_{i = 1}^{\eta'}[x_i \in C(v)]w(x_i, v) - 1 + \frac{s}{\eta_0}\sum_{i=1}^{\eta_0}(d_w(v, \{u_i\}) + d(v, \{u_i\}) - 1)$
        \item return $\rmCostStays - \rmCostMoves$
    \end{itemize}
\end{tcolorbox}

To compute $w(x_i,v)$, we only need to check whether $x_i$ and $v$ are in the same cluster in previous local search solution(s), which we can store in $O(n)$ space. Note that since $x_i$ is sampled from $N_G(v)$, we don't need to query whether it is a neighbor of $v$. Thus we can use Lemma \ref{lem:sample-neighbor} to sample $x_i$'s. In each round, since each vertex only samples a constant number of times, we only need a constant number of realizations of Lemma \ref{lem:sample-neighbor}.

It's a litter more tricky to sample $u_i$'s. When computing $d_w(v,\{u_i\})$, we need to query whether $v$ and $u_i$ are neighbors. Hence we need to have access to the neighborhood of $u_i$. We do the sampling as follows. We first use Lemma \ref{lem:sample-vertex} to sample some vertices. Let's call them special vertices. For the pivot vertex $r$ where we want to sample $u_i$'s from the admissible neighborhood $N(r)$, we know that with high probability, $\Omega(\log n)$ vertices in $N(r)$ are special, since $|N(r)|=\Omega(d(r))$ (when there exists the improving cluster) and vertices in $N(r)$ are degree-similar to $r$. Notice that vertices in $N(r)$ are sampled with similar but not exactly the same probabilities in Lemma \ref{lem:sample-vertex}, so we need to further discard each vertex with some (constant) probability to get independent samples. In the end, we sample $u_i$'s from these remaining special vertices. By Lemma \ref{lem:sample-vertex}, we can entirely store the neighborhood of all special vertices, so we can answer each neighbor query in $O(1)$ time. In each round, since each pivot only samples a constant number of times and different pivots sample from disjoint sets, we only need a constant number of realizations of Lemma \ref{lem:sample-vertex}.

Note that we do not need to store the $\eta'$ samples in the previous algorithms, since we can work with them one by one. So we do not introduce extra space cost here. The running time of this estimation is now $O_{\epsilon}(1)$.

\begin{lemma}
    Consider the setting of lemma~\ref{lem:concentration-weighted-noisy}, let $\eta'>\frac{W^2\eta^5}{4\epsilon^4}$, with probability $1 - 6\exp(-\frac{1}{2}\eta)$, the following two inequalities hold:
    \begin{equation}
        \label{eq:coststay-concentration-sampling}
        \rmEstCostStays_w(S,s,v) \in \rmExpCostStays_w(s,v) \pm \frac{1}{\eta^2}Ws
    \end{equation}
    \begin{equation}
        \label{eq:costmove-concentration-sampling}
        \rmEstCostMoves'_w(S,s,v) \in \rmExpCostMoves_w(s,v) \pm \left(\frac{1}{\eta^2}Ws + \frac{\epsilon^2}{\eta^2}d_w(v)\right)
    \end{equation}
\end{lemma}

If the lemma holds, we can replace $\rmEstCostMoves$ by $\rmEstCostMoves'$ in GenerateCluster, since they introduce the same asymptotic bound on the failure probability and error under the setting of lemma~\ref{lem:concentration-weighted-noisy}, the same proof still holds. 

\begin{proof}
    Let $Z_i = \frac{|N_G(v)|}{\eta'}[u_i \in C(v)]w(u_i, v)$. By definition, 
    \[
        \rmEstCostMoves'_w(S, s, v) - \rmEstCostMoves_w(S, s, v) = \sum_{i = 1}^{\eta'}Z_i - \E[\sum_{i = 1}^{\eta'}Z_i] ~.
    \]
    Since $Z_i \in [0, \frac{|N_G(v)|}{\eta'}W]$, according to Hoeffding's inequality, 
    \begin{align*}
         & \Pr\left[\left|\rmEstCostMoves'_w(S, s, v) - \rmEstCostMoves_w(S, s, v)\right| \geq \frac{\epsilon^2}{\eta^2}d_w(v)\right] \\
         & \qquad \leq 2\exp\left(-\frac{2\left(\frac{\epsilon^2}{\eta^2}d_w(v)\right)^2}{\eta'\left(\frac{|N_G(v)|}{\eta'}W\right)^2}\right) \\
         & \qquad \leq 2\exp\left(-\frac{2\eta'\epsilon^4d_w(v)^2}{\eta^4|N_G(v)|^2W^2}\right) \\
         & \qquad \leq 2\exp\left(-\frac{2\eta'\epsilon^4}{\eta^4W^2}\right) \\
         & \qquad \leq 2\exp\left(-\frac{1}{2}\eta\right)
    \end{align*}
    The second last line holds from the fact that $d_w(v) \geq |N_G(v)|$, and the last line holds from $\eta'>\frac{W^2\eta^5}{4\epsilon^4}$.
    By union bound, with probability $1 - 6\exp(-\frac{1}{2}\eta)$, the three inequalities, \eqref{eq:coststay-concentration-error}, \eqref{eq:costmove-concentration-error} and 
    \begin{equation}
        \label{eq:costmove-sampling-error}
        \left|\rmEstCostMoves'_w(S, s, v) - \rmEstCostMoves_w(S, s, v)\right| \leq \frac{\epsilon^2}{\eta^2}d_w(v)
    \end{equation}
    holds. Combining \eqref{eq:costmove-concentration-error} and \eqref{eq:costmove-sampling-error}, we get $\eqref{eq:costmove-concentration-sampling}$. This finishes the proof of this lemma. 
\end{proof}

\paragraph{Estimating the improvement.}
Then we show how to estimate the improvement, $\cost(\calC + S') - \cost(\calC)$, in $O_{\epsilon}(d(r))$ time, where $r$ is the selected vertex and $S' \subseteq N(r)$. 

\begin{tcolorbox}[beforeafter skip=10pt]
    $\rmEstImprovement_w(\calC, S', r)$
    \begin{itemize}[itemsep=0pt, topsep=4pt]
        \item Let $C(r) \in \calC$ be the cluster containing $r$
        \item $I \gets 0$
        \item For $j = 1, \dots, \eta'$:
        \begin{itemize}
            \item Uniformly randomly sample $u_j \in S' \Delta C(r)$ %
            \item $ I\gets  I + \rmCostStays_w(S', u_j) - \rmCostMoves_w(S', u_j)$\\ 
            \hphantom{1.1cm} $ - \frac{1}{2}\sum_{v \in S'}\left(\cost_w(\calC, u_jv) - \cost_w(\calC + S', u_jv)\right)$
            \\ where $\cost_w(\calC, u_jv)$ is the cost paid by the clustering scheme $\calC$ for the edge or non-edge $u_jv$. 
        \end{itemize}
        \item return $I|S' \Delta C(r)|/\eta'$
    \end{itemize}
\end{tcolorbox}

Here we also need to query the neighbors of $u_j$. Since $u_j$'s are sampled from $S'\Delta C(r)$ which also has a size $\Omega(d(r))$, we can use the same way to sample them as we previously sampling $u_i$'s.

We will prove the following lemma. Then, according to union bound, all the costs are estimated well within a small error if we set $\zeta=2\log n$. By changing the constant in the definition of $\Delta$ (in Section \ref{sec:number-of-rounds}), the same algorithm could also work. 

\begin{lemma}
    Let $\eta' > 38220595200\epsilon^{-54}W^2\zeta$ for some parameter $\zeta$. With probability at least $1 - 2\exp(-\zeta)$,
    \begin{equation}
        \label{eq:costfunction-concentration}
        \rmEstImprovement_w(C, S', r) \in \cost_w(\calC) - \cost_w(\calC + S') \pm \frac{\epsilon^{13}}{4}\cdot\frac{\epsilon^8}{576}|D(r)|\cdot|N(r)|.
    \end{equation}
\end{lemma}

\begin{proof}
    Let $\cost(\calC, uv)$ be the cost by the clustering scheme $\calC$ for the edge or nonedge $uv$. For each $u \in S'$, let 
    \begin{align*}
        X_u = & \frac{1}{2}\sum_{v \in S'}(\cost_w(\calC, uv) - \cost_w(\calC + S', uv)) + \\
        & \sum_{v \in C(u) \setminus S'}(\cost_w(\calC, uv) - \cost_w(\calC + S', uv))
    \end{align*}
    Observe that
    \begin{align*}
        X_u = & \frac{1}{2}\sum_{v \in S'}(\cost_w(\calC, uv) - \cost_w(\calC + S', uv)) + \\
        & \sum_{v \in C(u) \setminus S'}(\cost_w(\calC, uv) - \cost_w(\calC + S', uv)) \\
        = & \frac{1}{2}\sum_{v \in S'}(\cost_w(\calC, uv) - \cost_w(\calC + S', uv)) + \\
        & \sum_{v \in V \setminus S'}(\cost_w(\calC, uv) - \cost_w(\calC + S', uv)) \\
        = & \sum_{v \in V}(\cost_w(\calC, uv) - \cost_w(\calC + S', uv)) - \\
        & \frac{1}{2}\sum_{v \in S'}(\cost_w(\calC, uv) - \cost_w(\calC + S', uv)) \\
        = & \rmCostStays_w(S', u) - \rmCostMoves_w(S', u) - \frac{1}{2}\sum_{v \in S'}(\cost_w(\calC, uv) - \cost_w(\calC + S', uv))
    \end{align*}
    Note that $X_u$ can be computed in $O(d(u) + |S'|)$. 
    \begin{equation}
        \label{eq:split-improvement}
        \cost_w(\calC) - \cost_w(\calC + S') = \sum_{u \in S' \oplus C(r)}X_u
    \end{equation}
    Since $|X_u| \leq W(|S'| + |C(u)|) \leq W(|N(r)| + |C(u)|)$, by Corollary~\ref{cor:bounded-minus-cost} and lemma~\ref{lem:size-of-neighborhood}, we have
    \begin{align*}
        |X_u| & \leq W(|N(r)| + |C(u)|) \\
        & \leq W(6\epsilon^{-3}d(r) + 12\epsilon^{-4}d(u)) \\
        & \leq W(6\epsilon^{-3}d(r) + 24\epsilon^{-5}d(r)) \\
        & \leq 30W\epsilon^{-5}d(r)
    \end{align*}
    According to Hoeffding's inequality, 
    \begin{align*}
        & \Pr\left[\left|\frac{|S' \oplus C(r)|}{\eta'}\sum_{j = 1}^{\eta'}X_{u_j} - (\cost_w(\calC) - \cost_w(\calC + S'))\right| \geq \frac{\epsilon^{13}}{4}\cdot \frac{\epsilon^8}{576}|D(r)|\cdot|N(r)|\right] \\
        & \qquad \leq 2\exp\left(-\frac{2\left(\frac{\epsilon^{13}}{4}\cdot\frac{\epsilon^8}{576}|D(r)|\cdot|N(r)|\right)^2}{\eta'\left(\frac{|S'\oplus C(r)|}{\eta'}60W\epsilon^{-5}d(r)\right)^2}\right) \\
        & \qquad \leq 2\exp\left(-\frac{2\eta'\epsilon^{42}(|D(r)|\cdot|N(r)|)^2}{19110297600W^2\epsilon^{-10}|D(r)|^2d(r)^2}\right) \\
        & \qquad \leq 2\exp\left(-\frac{\eta'\epsilon^{44}}{38220595200W^2\epsilon^{-10}}\right) \\
        & \qquad \leq 2\exp\left(-\zeta\right)
    \end{align*}
    The last two inequality hold because $|N(r)| \geq \frac{1}{2}\epsilon d(r)$ and $\eta' > 38220595200\epsilon^{-54}W^2\zeta$. This finishes the proof of this lemma. 
\end{proof}

\section{Acknowledgements}
    David Rasmussen Lolck, Mikkel Thorup, Shuyi Yan and Hanwen Zhang are part of Basic Algorithms Research Copenhagen (BARC), supported by the VILLUM Foundation grant 16582. 
    Marcin Pilipczuk is partially supported by the VILLUM Foundation grant 16582 while visiting Basic Algorithms Research Copenhagen. 
    Hanwen Zhang is partially supported by Starting Grant 1054-00032B from the Independent Research Fund Denmark under the Sapere Aude research career programme. 
    \bibliographystyle{alpha}
    \bibliography{references}
\end{document}